\pdfoutput=1
\documentclass[a4paper,reprint,twocolumn,notitlepage,aip,cha,nofootinbib]{revtex4-2}
\usepackage[left=1.5cm,right=1.5cm,top=1cm,bottom=1.5cm,includeheadfoot]{geometry}

\usepackage[english]{babel}
\usepackage[T1]{fontenc}
\usepackage{tgtermes} 
\usepackage{tgheros} 
\renewcommand{\selectlanguage}[1]{} 

\usepackage{amsmath}
\usepackage{amssymb}
\usepackage{amsthm}
\usepackage{physics}
\usepackage{mathtools}

\usepackage{graphicx}
\usepackage[dvipsnames]{xcolor}
\usepackage{tikz}
\usepackage{tikz-3dplot} 

\usepackage{soul}
\usepackage{cancel}

\usepackage{hyperref}
\usepackage[capitalise]{cleveref}
\newcommand\numberthis{\addtocounter{equation}{1}\tag{\theequation}}

\usepackage{enumerate}
\usepackage[shortlabels]{enumitem}

\newtheorem{theorem}{Theorem}[section]
\newtheorem{lemma}[theorem]{Lemma}
\crefname{lemma}{Lemma}{Lemmas}
\Crefname{lemma}{Lemma}{Lemmas}
\newtheorem{corollary}[theorem]{Corollary}
\crefname{corollary}{Corollary}{Corollaries}
\Crefname{corollary}{Corollary}{Corollaries}
\newtheorem{proposition}[theorem]{Proposition}
\crefname{proposition}{Proposition}{Propositions}
\Crefname{proposition}{Proposition}{Propositions} 
\newtheorem{conjecture}[theorem]{Conjecture}
\crefname{conjecture}{Conjecture}{Conjectures}
\Crefname{conjecture}{Conjecture}{Conjectures} 

\theoremstyle{remark}

\crefname{remark}{Remark}{Remarks}
\Crefname{remark}{Remark}{Remarks}

\newtheorem{example}{Example}[section]
\crefname{example}{Example}{Examples}
\Crefname{example}{Example}{Examples}

\theoremstyle{definition}
\newtheorem{definition}[theorem]{Definition}
\crefname{definition}{Def.}{Defs.}
\Crefname{definition}{Definition}{Definitions}

\newcommand{\crefpart}[2]{%
  \namecref{#1}~\hyperref[#2]{\labelcref*{#1}.\ref*{#2}}%
}

\newcommand{\RR}{\mathbb{R}}
\newcommand{\CC}{\mathbb{C}}
\newcommand{\dua}[2]{\langle #1, #2 \rangle}

\DeclareMathAlphabet{\mathmybb}{U}{bbold}{m}{n}
\newcommand{\id}{\mathmybb{1}}
\newcommand{\vsigma}{\boldsymbol{\sigma}}

\newcommand{\vv}{\boldsymbol{v}}
\renewcommand{\time}{\tau}
\newcommand{\dtime}{\frac{\dd}{\dd\time}}
\newcommand{\ddtime}{\frac{\dd^2}{\dd\time^2}}
\newcommand{\s}{\mathrm{s}}
\newcommand{\vdxc}{v_\mathrm{dxc}}

\newcommand{\vd}{v_\mathrm{d}}
\newcommand{\vdx}{v_\mathrm{dx}}

\newcommand{\vdc}{v_\mathrm{dc}}
\newcommand{\vx}{v_\mathrm{x}}

\newcommand{\pf}{{\mathrm{pf}}}

\newcommand{\qop}{\hat{q}}
\newcommand{\qval}{q}

\renewcommand{\norm}[1]{\| #1 \|}
\renewcommand{\abs}[1]{| #1 |}
\newcommand{\Abs}[1]{\left| #1 \right|}

\begin{document}
\author{Vebjørn H. Bakkestuen}
\affiliation{Department of Computer Science, Oslo Metropolitan University, Oslo, Norway}
\author{Vegard Falmår}
\affiliation{Department of Computer Science, Oslo Metropolitan University, Oslo, Norway}
\author{Maryam Lotfigolian}
\affiliation{Department of Computer Science, Oslo Metropolitan University, Oslo, Norway}
\author{Markus Penz}
\affiliation{Department of Computer Science, Oslo Metropolitan University, Oslo, Norway}
\affiliation{Max Planck Institute for the Structure and Dynamics of Matter and Center for Free-Electron Laser Science, Hamburg, Germany}
\author{Michael Ruggenthaler}
\affiliation{Max Planck Institute for the Structure and Dynamics of Matter and Center for Free-Electron Laser Science, Hamburg, Germany}
\affiliation{The Hamburg Center for Ultrafast Imaging, Hamburg, Germany}
\author{Andre Laestadius}
\affiliation{Department of Computer Science, Oslo Metropolitan University, Oslo, Norway}
\affiliation{Hylleraas Centre for Quantum Molecular Sciences, Department of Chemistry, University of Oslo, Oslo, Norway}
\email{andre.laestadius@oslomet.no}

\title[QEDFT: Quantum Rabi Model]{Quantum-Electrodynamical Density-Functional Theory\\Exemplified by the Quantum Rabi Model}

        \begin{abstract}
            The key features of density-functional theory (DFT) within a minimalistic implementation of quantum electrodynamics are demonstrated, thus allowing to study elementary properties of quantum-electrodynamical density-functional theory (QEDFT). We primarily employ the quantum Rabi model, that describes a two-level system coupled to a single photon mode, and also discuss the Dicke model, where multiple two-level systems couple to the same photon mode. In these settings, the density variables of the system are the polarization and the displacement of the photon field. We give analytical expressions for the constrained-search functional and the exchange-correlation potential and compare to established results from QEDFT. We further derive a form for the adiabatic connection that is almost explicit in the density variables, up to only a non-explicit correlation term that gets bounded both analytically and numerically. This allows several key features of DFT to be studied without approximations.\\[.75em]            
			This paper is a contribution to the ``Trygve Helgaker Festschrift'' on the occasion of Trygve's 70th birthday. Some of the authors had the privilege of having Trygve as a teacher and mentor in the field of mathematical DFT. It is therefore with great honor that we dedicate this work to the celebration of his birthday.
        \end{abstract}

\maketitle

\tableofcontents

\section{Introduction}\label{sec:Intro}
\subsection{Prelude and Overview}
The study of light-matter interactions forms the basis for understanding a wide range of phenomena whose effects are instrumental for measuring and manipulating matter in experiments.
At the fundamental level, charged particles interact among each other through their coupling to the photon field, a process that is described by quantum electrodynamics (QED)~\cite{Weinberg_1995,Peskin_1995,MandlShaw,ryder1996quantum,greiner2013field,spohn2004dynamics}. While the quantization of the electromagnetic field is often considered to only be relevant for high-energy physics, QED effects, such as spontaneous emission or the Purcell effect, also occur in the low-energy (non-relativistic) regime of charged particles. In recent years, many experimental and theoretical works have shown that in optical environments, such as Fabry--Pérot cavities, changes in the quantized light field can modify chemical and material properties even at equilibrium~\cite{ebbesen2016hybrid,garcia2021manipulating,Ruggenthaler2023}. It has therefore become increasingly relevant to extend well-established first-principles methods, such as density-functional theory (DFT) and coupled-cluster theory, to encompass QED~\cite{Ruggenthaler2014,Ruggenthaler2017,Ruggenthaler2018,haugland2020coupled,mordovina2020polaritonic}.

Due to its computational simplicity, DFT is the method predominantly used when studying quantum systems with large numbers of particles. In DFT, the $N$-body wavefunction---with its intractable dimensionality---is replaced by the one-body particle density. This dimensional reduction is precisely why DFT calculations provide effective approximations and thus have become an indispensable tool across many fields, such as chemistry, materials science and solid-state physics~\cite{Burke2012,Verma_2020,Teale2022}. Since the seminal papers of \citet{Hohenberg1964}, \citet{KS1965}, and \citet{Lieb1983}, significant efforts have been devoted both to the numerical and mathematical developments of DFT. Besides, motivated by its extremely elegant formulation in terms of convex analysis, a perspective towards DFT has emerged that makes it more than just an approximation method. The concave form of the ground-state energy in terms of the external potential naturally yields the universal functional as its Legendre--Fenchel transform. Then the Hohenberg--Kohn theorem is the statement that the subdifferential of the universal functional just contains a single element (see also \cref{sec:DFTBackground}) and $v$-representability connects closely to differentiability of this functional. The whole of DFT thus follows as a convex treatment of many-body quantum mechanics in the ground state. In this sense, DFT can be referred to more as a discovery than an invention~\cite{Helgaker-personal}. The DFT formulation of QED presented here serves a case in point for this viewpoint.

In this article we will exemplify several key features of DFT by derivations and examples using simple models for QED---settings that allow significantly more explicit constructions than the standard Coulombic DFT. This means the whole formalism of DFT can be well defined and important results like the Hohenberg--Kohn theorem or $v$-representability receive a mathematical rigorous treatment. Apart from formulating the basics of how to build up QEDFT, this has the additional value of constituting a pedagogical showcase for DFT in general. The richness of mathematical and physical concepts that enter the theory will become clearly appreciable in the course of this work. In this sense, this complements our previous, more technical work that aimed at defining the basic DFT structures for the Dicke model~\cite{Bakkestuen2024}.

The article is structured as follows. This introductory section continues with a short overview of quantum-electrodynamical density-functional theory (QEDFT) and briefly explain the landscape and hierarchy of QED models surrounding our model of choice, the \textit{quantum Rabi model}. We conclude the section with a summary of important concepts in the mathematical formulation of standard DFT. The quantum Rabi model setting, including the Hamiltonian, general ground-state properties, and hyper-virial results are presented in \cref{sec:QRabiModel}. Note that \cref{subsec:DickeModel,subsec:HKDickeModel} pertain to a generalization to the Dicke model. In a first reading these sections may be safely skipped as they do not directly affect the analysis of the quantum Rabi model, however they illustrate important aspects of DFT not readily seen in the simpler setting. \Cref{sec:HK} introduces the polarization and displacement as internal variables of this QEDFT formulation and establishes the corresponding Hohenberg--Kohn-type results. The DFT theory part is then continued by defining the Levy--Lieb constrained-search functional and studying its properties for the model at hand in \Cref{sec:LL}. This also establishes the important result that all internal-variable pairs are in fact $v$-representable if they are not on the boundary of their domain. As an important tool of DFT the adiabatic connection is analyzed in \Cref{sec:AC}, where we can give an almost explicit form due to the relative simplicity of the model. Finally, functional approximation that are based on a photon-free formulation are discussed in \Cref{sec:pf} and we compare to other results from QEDFT. We conclude in \Cref{sec:Conclusions}.

\subsection{Quantum-Electrodynamical Density-Functional Theory}\label{subsec:QEDFT}
Of specific interest for this work is the extension of DFT to QED, which has been termed QEDFT and has been investigated for the static~\cite{Ruggenthaler2017, penz2023structure2} as well as the time-dependent (and even relativistic) case~\cite{ruggenthaler2011time,tokatly2013time,Ruggenthaler2014,flick2019light,jestadt2019light,konecny2024relativistic}. As a formulation, QEDFT has been applied to a broad variety of physical and chemical situations and different approximations to the new photon-matter exchange-correlation field in the Kohn--Sham formulation of QEDFT have been proven to provide accurate results~\cite{pellegrini2015optimized,Flick2018,yang2021quantum,flick2022simple,liebenthal2023assessing,lu2024electron,dengoptical}. Yet a detailed mathematical investigation of this new form of a density-functional reformulation, similar to how it was performed in standard DFT~\cite{Lieb1983,eschrig2003,Helgaker2022bookchap}, has so far not been pursued, apart from \citet{Bakkestuen2024} that now led to this work. Such an investigation is, however, important not only as fundamental question, but also to further guide the development of QEDFT and its approximation strategies.

While most of the rigorous considerations in QEDFT are based on the Pauli--Fierz Hamiltonian~\cite{spohn2004dynamics}, various approximations to this Hamiltonian are used as a starting point for further investigations~\cite{Ruggenthaler2023}. These approximate Hamiltonians lead to a hierarchy of QEDFTs~\cite{Ruggenthaler2014,jestadt2019light} and yield a connection to well-established models of quantum optics that are designed to describe the photonic subsystem well, while significantly simplifying the matter part. One such paradigmatic quantum-optical model is the quantum Rabi model~\cite{Xie2017}, which will be the main focus of this article. Given the vast number of models available, let us briefly orient ourselves on where in the landscape of QED models our investigation will take place, before turning to the analysis.

\subsection{Models in QED}\label{subsec:QEDModels}
QED is arguably the most accurate description of light-matter interactions and as a field it is as old as quantum mechanics itself~\cite{Weinberg_1995}. It is described by the QED Lagrangian $\mathcal{L}_\mathrm{QED}$, and its equations of motion are the relativistic Dirac equation and the wave equation for the vector potential~\cite{MandlShaw,Peskin_1995,Weinberg_1995}. Although QED is one of the most accurate descriptions of nature ever conceived and has been studied for almost 100 years, it is full of challenges. Notably, within the full field-theoretic treatment QED processes are calculable only perturbatively and the complexity of the calculations rapidly increases in the perturbative expansion. Moreover, most applications in chemistry and solid-state physics involve energies that are not sufficiently large that relativistic effects become important, at least not for the fermionic degrees of freedom. In these applications, the treatment of the fully relativistic Dirac equation thus becomes superfluous. Within low-energy applications of QED one usually chooses the Coulomb gauge, also known as the radiation or transversal gauge, as it has a number of useful properties, especially in the semi-classical regime (quantum particles and classical radiation fields). In particular it removes the unphysical degrees of freedom in the gauge field, ensuring that the vector potential only has the two transversal polarizations, and it picks out the Coulomb potential to describe interactions between the fermions.

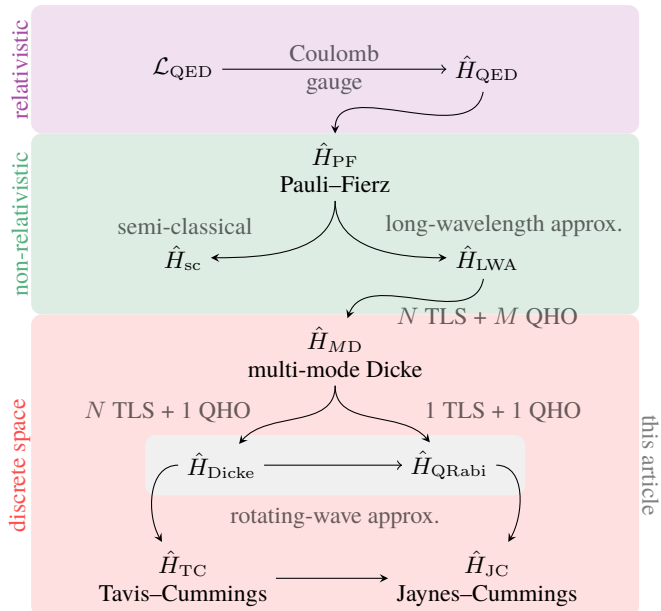
\begin{figure}[tb]
    \centering
\usetikzlibrary{shapes.geometric, arrows,positioning}
\tikzstyle{VioletBlob} = [rectangle, rounded corners, minimum width=8cm, minimum height=1.7cm,text centered, draw=none, fill=violet!12]
\tikzstyle{GreenBlob} = [rectangle, rounded corners, minimum width=8cm, minimum height=2.4cm,text centered, draw=none, fill=ForestGreen!12]
\tikzstyle{RedBlob} = [rectangle, rounded corners, minimum width=8cm, minimum height=4cm,text centered, draw=none, thick, fill=red!12]
\tikzstyle{GrayBlob} = [rectangle, rounded corners, minimum width=5cm, minimum height=0.8cm,text centered, draw=none, fill=gray!12]

\begin{tikzpicture}[node distance=2cm]
    \node [VioletBlob] (rel) {};
    \node (nonrel) [GreenBlob, below of=rel, node distance=2.051cm] {}; 
    \node (discrete) [RedBlob,below of=nonrel,node distance=3.201cm] {};
    
    \node (RelLabel) [left of=rel, node distance=4.15cm,rotate=90,opacity=1] {\textcolor{violet!70}{relativistic}};
    \node (NonRelLabel) [left of=nonrel, node distance=4.15cm,rotate=90,opacity=1] {\textcolor{ForestGreen!70}{non-relativistic}};
    \node (DiscreteLabel) [left of=discrete, node distance=4.15cm,rotate=90,opacity=1] {\textcolor{red!70}{discrete space}};
    
    \node (LQED) [left of=rel,node distance=2cm] {$\mathcal{L}_\mathrm{QED}$};
    \node (HQED) [right of=rel,node distance=2cm] {$\hat{H}_\mathrm{QED}$};

    \node (PauliFierz) [below of=rel,node distance=1.3cm,align=center] {$\hat{H}_\mathrm{PF}$ \\ \small{Pauli--Fierz}};

    \node (nonrel1) [below of=PauliFierz, node distance=1.2cm] {};
    \node (PFLWA) [right of=nonrel1,node distance=2cm,align=center] {$\hat{H}_\mathrm{LWA}$};
    \node (SemiClassical) [left of=nonrel1,node distance=2cm,align=center] {$\hat{H}_\mathrm{sc}$};
     
    \node (MMDExplain) [below of=PFLWA, node distance=0.8cm,opacity=0.6] {\small{$N$ TLS + $M$  QHO}};

    \node (MultiDicke) [below of=nonrel1, node distance=1.25cm, align=center] {$\hat{H}_{M\mathrm{D}}$\\ \small{multi-mode Dicke}};

    \node (ModelLevel) [GrayBlob,below of=MultiDicke,node distance=1.5cm] {};
    \node (ThisLabel) [right of=ModelLevel, node distance=4.15cm,rotate=270,opacity=1] {\textcolor{gray}{this article}};
    
    \node (Dicke) [left of=ModelLevel,node distance=1.5cm] {$\hat{H}_\mathrm{Dicke}$};
    \node (Rabi) [right of=ModelLevel,node distance=1.5cm] {$\hat{H}_\mathrm{QRabi}$};
    \node (RWA) [below of=ModelLevel,node distance=0.7cm,opacity=0.6]  {\small{rotating-wave approx.}};

    \node (LWAExplain) [above right of=PFLWA, node distance=0.6cm, xshift=-0.2cm, opacity=0.6] {\small{long-wavelength approx.}};
    \node (scExplain) [above right of=SemiClassical, node distance=0.6cm, xshift=-0.4cm, opacity=0.6] {\small{semi-classical}};
    \node (DickeExplain) [above left of=Dicke,node distance=1cm, opacity=0.6] {\small{$N$ TLS + 1 QHO}};
    \node (RabiExplain) [above right of=Rabi,node distance=1cm,opacity=0.6] {\small{1 TLS + 1 QHO}};

    \node (Cummings) [below of=ModelLevel, node distance=1.5cm] {};    
    \node (JaynesCummings) [right of=Cummings,node distance=2cm,align=center] {$\hat{H}_\mathrm{JC}$ \\ \small{Jaynes--Cummings}};
    \node (TavisCummings) [left of=Cummings,node distance=2cm,align=center] {$\hat{H}_\mathrm{TC}$ \\ \small{Tavis--Cummings}};
    
    \draw[->] (LQED) -- (HQED) node[midway,opacity=0.6,rotate=0,align=center] {\small{Coulomb}\\ \small{gauge}};
    \draw[->,>=stealth] (HQED) to[out=260,in=90] (PauliFierz);
    \draw[->,>=stealth] (PauliFierz) to[out=270,in=0] (SemiClassical);
    \draw[->,>=stealth] (PauliFierz) to[out=270,in=180] (PFLWA);
    \draw[->,>=stealth] (PFLWA) to[out=260,in=70] (MultiDicke);
    \draw[->,>=stealth] (MultiDicke) to[out=270,in=50] (Dicke);
    \draw[->,>=stealth] (MultiDicke) to[out=270,in=130] (Rabi);
    \draw[->] (Dicke) -- (Rabi) {};
    \draw[->,>=stealth] (Dicke) to[out=180,in=120] (TavisCummings) ;
    \draw[->,>=stealth] (Rabi) to[out=0,in=60] (JaynesCummings) {};
    \draw[->,>=stealth] (TavisCummings) to[out=0,in=180]  (JaynesCummings) {};
\end{tikzpicture}
    
    \caption{Hierarchy of QED models relevant for the current analysis. The top panel (violet) contains the fully relativistic theory, the middle panel (green) is the non-relativistic limit in form of the Pauli--Fierz Hamiltonian and further simplifications, and bottom panel (red) has different models following from discretizations. The models of interest in this article are marked by the injected gray panel, of particular interest is the quantum Rabi model.}
    \label{fig:ModelHierarchy}
\end{figure}

The wide range of possible simplifications of the theory gives rise to a large hierarchy of models. A small snapshot of such models is outlined in \cref{fig:ModelHierarchy}. Let us briefly comment on how they connect to each other. Starting from $\mathcal{L}_\mathrm{QED}$, there are a multitude of gauge-fixing conditions, each with a particular set of advantages and disadvantages. By choosing the Coulomb gauge and equal-time commutation relations, one arrives at what we will refer to as the \textit{relativistic QED Hamiltonian} $\hat{H}_\mathrm{QED}$ (as the volume integral over the QED Hamiltonian density). Then, in taking the non-relativistic (low energy) limit, one arrives at the \textit{Pauli--Fierz Hamiltonian} $\hat{H}_\mathrm{PF}$~\cite{spohn2004dynamics}, where usually the Born--Oppenheimer approximation is already included. A further important simplification is the long-wavelength or dipole approximation, where the transfer of momentum between light and matter is assumed to be zero~\cite{Ruggenthaler2023}. In this context, we note that different forms of the resulting Hamiltonian are possible.  In certain forms besides the Coulomb matter-matter interaction, photon-induced direct matter-matter coupling terms appear, which are called dipole-self-energies or self-polarization terms~\cite{rokaj2018light, schaefer2020relevance}. In the following we discard such direct matter-matter interaction terms that arise due to further transformations of the dipole-approximated Pauli--Fierz Hamiltonian but want to highlight that these terms can become important in certain situations~\cite{sidler2024unraveling}. Despite simplifications, such models are still difficult to solve. Thus, many applications rely on further approximations and one possible avenue is to describe the light part classically while keeping the matter part quantum, which yields what is usually referred to as a \textit{semi-classical} approach. An example of such approach is the WKB (Wentzel--Kramers--Brillouin) approximation. On the other side, the models we are interested in here use discretization in both the light and matter parts, but retain the ``quantumness'' of both. For such models, one picks out $M$ photonic modes from the expansion of the vector-potential operator and discards the rest, an efficient approach when only a limited number of photonic modes actively couple to the problem. This allows the photonic sector to be treated as $M$ quantum harmonic oscillators (QHO). If one further simplifies the setting by treating the fermionic sector as $N$ two-level systems (TLS), one obtains what we refer to as the \textit{multi-mode Dicke Hamiltonian} $\hat{H}_\mathrm{MD}$. 

Of particular interest for this article are two variants of the multi-mode Dicke model. In particular, the restriction to just one two-level system and one photonic mode, known as the \textit{quantum Rabi model} in the literature, whose Hamiltonian $\hat{H}_\mathrm{QRabi}$ is given in \cref{eq:InternalHamiltonian}. The other one is the slightly more general \textit{Dicke model}, whose Hamiltonian $\hat{H}_\mathrm{Dicke}$ is given by \cref{eq:DickeHamiltonian}. The model consists of  $N$ two-level systems and a single QHO, and will be studied in selected parts of the article. These models have also gained interest in the mathematics community due to their quite intricate spectral properties~\cite{Braak2011,Braak2013-Dicke,Braak2015}. Recently, a quite advanced DFT approximation for the Dicke model was suggested by \citet{Novokreschenov2023}. Finally, by performing the rotating-wave approximation, the Dicke and quantum Rabi models can be reduced to the Tavis--Cummings and Jaynes--Cummings models, respectively.

\subsection{Elements from Standard Density-Functional Theory}\label{sec:DFTBackground}
Let us briefly recapitulate the mathematical description of standard DFT, in order to set the stage for our analysis of QEDFT. Standard DFT uses the one-body particle density $\rho$ as the primary quantity to describe a system of $N$ electrons within a quantum-mechanical treatment. Here, the density is defined from the wavefunction as
\begin{equation*}
    \rho_\psi (\mathbf{x}) = N  \int_{\RR^{3(N-1)}} \abs{\psi(\mathbf{x},\mathbf{x}_2, \dots, \mathbf{x}_N)}^2 \dd{\mathbf{x}_2} \dots \dd{\mathbf{x}_N} .
\end{equation*}
The set of fermionic $N$-particle wavefunctions considered are $L^2$-normalized and have finite kinetic energy, which gives a set of $N$-representable densities $\mathcal I_N$~\cite{Lieb1983} defined by the following properties,
\begin{equation*}
    \begin{aligned}
        & \rho(\mathbf{x}) \geq 0,\quad  \int_{\RR^3} \rho(\mathbf{x}) \dd \mathbf{x} =N,\\
        & T_\mathrm{vW}(\rho) = \frac 1 2 \int_{\RR^3} \big\vert \nabla \sqrt{\rho(\mathbf{x})} \big\vert^2 \dd \mathbf{x}< \infty.
    \end{aligned}
\end{equation*}
The set of $v$-representable densities is then the subset of densities $\rho_\psi$ originating from ground states $\psi$ of the Schrödinger equation with all possible external potentials $v$.\footnote{This way, the set of $v$-representable densities depends on the selected set of allowed potentials. A natural choice is $L^{3/2} + L^\infty$ as explained later in this section.} The result that each $v$-representable density then corresponds to exactly one potential $v$ (up to a constant) is the celebrated Hohenberg--Kohn theorem~\cite{Hohenberg1964}. Strictly speaking, further restrictions on the class of potentials considered are needed, before we can conclude that the potential is unique. We point the interested Reader to \citet{Garrigue2018} for further details. Perhaps even more important, the Hohenberg--Kohn theorem does not say anything about a possible surjectivity of the mapping from potentials to densities, i.e., it does not tell us which densities are $v$-representable in the first place.

Lieb's~\cite{Lieb1983} explicit construction of a wavefunction from a determinant shows that for each $\rho\in\mathcal I_N$ there is an anti-symmetric wavefunction $\psi$ that has $\rho_\psi =\rho$ and finite kinetic energy. This is what the term $N$-representability for densities refers to. Then, Levy's~\cite{Levy79} direct way of transitioning from a variation over wavefunctions to instead obtain the ground-state energy by means of a constrained search introduces the Levy--Lieb functional $F_\mathrm{LL}$,
\begin{align*}
        E(v) &= \inf_\psi \qty{ \langle\psi | \hat H(v) | \psi \rangle : \norm{\psi} = 1,\, \norm{\nabla\psi} < \infty } \\
        & = \inf_\rho \Big\{ \inf \qty{ \langle\psi | \hat{T} + \hat{W} | \psi \rangle : \psi\in \mathcal M_\rho } + \langle v,\rho \rangle \Big\} \\
        & = \inf_\rho \Big\{ F_\mathrm{LL}(\rho) + \langle v,\rho \rangle \Big\}.\numberthis{\label{eq:Egs}}
\end{align*}
Here, $\hat H(v)=\hat{T} + \hat{W} +\sum_{i=1}^N v(\mathbf{x}_i)$, with $\hat{T}$ and $\hat{W}$ the standard kinetic energy and two-body operators, respectively, $\langle v,\rho \rangle = \int_{\mathbb R^3} v\rho$ and we introduced the constraint manifold
\begin{equation*}
    \mathcal M_\rho =\qty{ \psi : \norm{\psi} = 1,\, \norm{\nabla\psi} < \infty,\, \psi \mapsto \rho }.
\end{equation*}
\Cref{eq:Egs} above defines the Levy--Lieb functional, 
\begin{equation*}
    F_\mathrm{LL}(\rho) = \inf \qty{ \langle\psi | \hat{T} + \hat{W}  | \psi \rangle : \psi\in \mathcal M_\rho} ,
\end{equation*}
as a constrained minimization of the internal energy over $\mathcal M_\rho$. Its effective domain (the points where it is finite) is the set of $N$-representable densities $\mathcal I_N$. It is worth pointing out that even for $N$-representable densities, an optimizer $\psi_0\in \mathcal M_\rho$ might not be associated with a Lagrange multiplier. In standard DFT, this Lagrange multiplier would be a scalar potential for which $\psi_0$ is a ground state (or even an excited state). We bring this point up here for the benefit of the Reader to contrast with the results available for the quantum Rabi and Dicke models of QEDFT presented later. 

An alternative perspective on DFT is to view the ground-state energy and universal density functional as a conjugate (Legendre-)pair $(F,E)$. In this spirit, Lieb introduced the \emph{convex} functional
\begin{equation*}
    F(\rho) = \sup_v \Big\{ E(v) - \langle v,\rho \rangle \Big\},
\end{equation*}
where $E(v)$ is the ground-state energy introduced in \cref{eq:Egs}. It holds that $F$ is the convex envelope of $F_\mathrm{LL}$ and that it is equal to the constrained-search functional over mixed (instead of pure) states~\cite{Lieb1983}. We can further replace $F_\mathrm{LL}$ by $F$ in the expression for the ground-state energy, \cref{eq:Egs}. Employing an unusual sign convention\footnote{Instead of the standard convex conjugate $f^*(y) = \sup_{x}(\langle y,x\rangle - f(x))$, we use $E(v) = F^\wedge(v) = -F^*(-v)$ and $F(\rho) = E^\vee(\rho) = (-E)^*(-\rho)$.}, $(F,E)$ forms a conjugate pair, where a ground-state density $\rho$ together with its potential $v$ saturates the Fenchel--Young inequality $E(v) - F(\rho) \leq \langle v,\rho\rangle$, i.e., $E(v) - F(\rho) = \langle v,\rho\rangle$.  This \emph{Lieb functional} $F$ is defined on $X=L^1\cap L^3$, which properly contains $\mathcal I_N$, where $F(\rho)=\infty$ for all $\rho$ that are not $N$-representable. In this mathematical setting, the natural space of potentials becomes $X^*=L^{3/2} + L^\infty$, which includes molecular Coulomb-type potentials $v(\mathbf x) =\sum_a Z_a |\mathbf R_a - \mathbf x|^{-1}$. The potential space $X^*$ is the dual space of $X$ and ensures finite interaction between the density $\rho$ and the potential $v$ by Hölder's inequality,
\begin{equation*}
    \left\vert \langle v,\rho \rangle \right\vert \leq \Vert v \Vert_{X^*} \Vert \rho\Vert_X . 
\end{equation*}

Placing our attention again on the variational principle for the ground-state energy, \cref{eq:Egs}, we note that since the variation of $\rho$ at any stationary point must be zero, we can equally formulate this with a differential. Then a potential $v$ determined like this actually yields $\rho$ in the ground state if we happen to be in a global minimum of $\rho\mapsto F_\mathrm{LL}(\rho) + \langle v,\rho \rangle$, else it can still be the density of an excited state. But remember that we have $F(\rho)$ convex, so using this functional instead, we can be sure that we obtain a potential that actually yields the correct density in the ground state. We thus write
\begin{equation}\label{eq:F-subdiff}
    -v\in\underline{\partial} F(\rho)
\end{equation}
where we used the generalized concept of a \emph{subdifferential}, because $F(\rho)$ is not differentiable in the usual sense~\cite{Lammert2007}. The subdifferential is defined as the set of all bounded tangent functionals below $F$ at $\rho$,
\begin{equation*}
\begin{aligned}
    & \underline{\partial} F(\rho) \\ 
    &= \big\{  v\in X^* \mid \forall\rho'\in X: F(\rho')
    \geq \underbrace{F(\rho)+\langle v,\rho'-\rho \rangle}_{\text{tangent of $F$ at $\rho$}} \big\}.
\end{aligned}
\end{equation*}
That is, if the subdifferential is non-empty then the potential is an element, else the density must be marked as non-$v$-representable. In standard DFT the potential is always unique by the Hohenberg--Kohn theorem (\emph{if} it exists), so we have at most one element in the subdifferential $\underline{\partial} F(\rho)$ (up to adding a constant). This, however, is no longer true in other variants of DFT. Different potentials can lead to the same density variables and thus provide counterexamples to the Hohenberg--Kohn theorem on a finite lattice~\cite{penz2021-Graph-DFT} or in paramagnetic current-DFT~\cite{Capelle2002,Tellgren2012,LaestadiusBenedicks2014}, while the situation remains undecided in total-current DFT~\cite{Laestadius2021}.

Finally, we address the important exchange-correlation contribution $E_\mathrm{xc}(\rho)$ to $F(\rho)$. This can be given by the \emph{density-fixed adiabatic connection}~\cite{LANGRETH19751425,Polak-et-al-2024} that introduces a coupling constant $\lambda$ in front of $\hat{W}$. It is then possible to connect the simpler non-interacting case $\lambda=0$ to the physical system of full interactions at $\lambda=1$ (as dictated by $\hat W$), or beyond with $\lambda >1$ (or even taking $\lambda \to \infty$, referred to as the strong-interaction limit~\cite{PhysRevA.60.4387}). For a fixed $\rho$, this defines a $\lambda$-dependent Lieb functional $F^\lambda(\rho)$. We can then connect this functional at any $\lambda>0$ to the non-interacting (kinetic energy only) $T(\rho)=F^0(\rho)$ by means of an integral representation with the Newton--Leibniz trick~\cite{Laestadius2024}
\begin{equation}\label{eq:NewtonLeibniz}
    F^\lambda(\rho) = T(\rho) + \int_0^\lambda f^\nu(\rho) \dd{\nu}.
\end{equation}
Here we have used that $\lambda \mapsto F^\lambda(\rho)$ is a concave function and where $f^\nu(\rho)$ is any element of the \emph{superdifferential} of $\lambda\mapsto F^\lambda(\rho)$. The superdifferential is the collection of all tangents that lie above the function, and is the concave equivalent of the subdifferential that we defined above. More precisely, in the current context the superdifferential is the set-valued mapping given by 
\begin{equation}\label{eq:F-superdiff}
    \overline{\partial}_\lambda F^\lambda = \qty{f \in \RR \mid \forall \lambda'\in \RR: F^{\lambda'} \leq F^{\lambda}+ f\cdot( \lambda'-\lambda ) }.
\end{equation}
We note in passing that the superdifferential is given by the interval $[a,b]$ with $a$ being the right and $b$ the left derivative of $F^\lambda$ at $\lambda$ (and that those always exist and are equal except at countably many points). Assuming that the density matrix $\Gamma^\lambda$ minimizes the internal energy from $\hat T + \lambda \hat W$ under the given density constraint, it can be pointed out that~\cite{Laestadius2024}
\begin{equation}\label{eq:F-superdiff-trace-W}
    f^\lambda(\rho) = \Tr \hat W \Gamma^\lambda  \in \overline{\partial}_\lambda F^\lambda(\rho).
\end{equation}
Furthermore, we can subtract the Hartree energy, $E_\mathrm{H}(\rho) = \iint \rho(\mathbf{r})\rho(\mathbf{r}')/\abs{\mathbf{r} - \mathbf{r}'} \dd[3]{\mathbf{r}} \dd[3]{\mathbf{r}'}$, such that
\begin{equation}\label{eq:ACDFT}
    \begin{aligned}
        F^\lambda(\rho) &= T(\rho) + \lambda E_\mathrm{H}(\rho) + \lambda E_\mathrm{xc}^\lambda(\rho), \\
        E_\mathrm{xc}^\lambda(\rho) &= \frac 1 \lambda \int_0^\lambda (f^\nu(\rho) - E_\mathrm{H}(\rho))\dd \nu.
    \end{aligned}
\end{equation}
If we set $ E_\mathrm{xc}(\rho)= E_\mathrm{xc}^{\lambda=1}(\rho)$, we see that the adiabatic connection allows us to write (recall that $\Gamma^\lambda$ is the minimizer corresponding to $F^\lambda(\rho)$)
\begin{equation*}
    E_\mathrm{xc}(\rho) =  \int_0^1 ( \Tr \hat W \Gamma^\lambda - E_\mathrm{H}(\rho))\dd \lambda,
\end{equation*}
which is an important result in exact DFT for the unknown exchange-correlation contribution to 
$F(\rho)$.

Furthermore, the term $E_\mathrm{xc}$ can itself be partitioned. The Levy--Perdew definition~\cite{LevyPerdewIJQC1994} of the exchange energy as the high-density limit $E_\mathrm{x}(\rho) = \lim_{\gamma \to \infty} \gamma^{-1}E_\mathrm{xc}(\rho_\gamma)$ (where $\rho_\gamma(\mathbf{r}) = \gamma^3 \rho(\gamma \mathbf r)$), then implies~\cite{Laestadius2024} (from the fact that $F^{\gamma\lambda}(\rho_\gamma) = \gamma^2 F^\lambda(\rho)$)
\begin{equation}\label{eq:DefExhange}
    E_\mathrm{x}(\rho) := \lim_{\lambda\to 0^+}\frac{F^\lambda (\rho) - F^0(\rho)}{\lambda}
    - E_\mathrm{H}(\rho).
\end{equation}
In other words, the exchange energy is (after subtracting $E_\mathrm{H}$) the right derivative of the mapping $\lambda\mapsto F^\lambda(\rho)$ at $\lambda=0$. The correlation term $E_\mathrm{c}$ is then simply the difference $E_\mathrm{xc} - E_\mathrm{x}$.

Even from this short summary it becomes obvious that DFT has a very rich mathematical structure, but it is also ripe with difficulties such as non-$v$-representability and non-differentiability~\cite{Lammert2007,penz2023structure1}. That is why it is beneficial, especially when including new effects, to study the extended theory thoroughly on the basis of simple model systems. It is thus the objective of this work to detail the above mathematical formulation of standard DFT in the context of model systems for QEDFT. Let us first turn to the QED-model Hamiltonians at hand, i.e., the quantum Rabi and Dicke models.

\section{The Quantum Rabi Model}\label{sec:QRabiModel}
\subsection{Model Definition}

In atomic units\footnote{That is units such that $\hbar = m_e = e = 4 \pi \varepsilon_0 = 1$.}, the Hamiltonian of the quantum Rabi model, $\hat{H}_\mathrm{QRabi}$, from now on denoted $\hat{H}_0$,  can be written as
\begin{equation}\label{eq:InternalHamiltonian}
    \hat{H}_0 = \frac{1}{2} \hat{p}^2 + \frac{\omega^2}{2} \qop^2
        - t \hat{\sigma}_x + g \hat{\sigma}_z \qop.
\end{equation}
It describes a single two-level system, like two levels in an atom or molecule, a simple dipole, or the spin states of an electron. This ``matter'' system gets coupled to a single photonic mode modeled as a quantum harmonic oscillator (QHO). Here $\omega$ is the frequency of the photonic mode, $t > 0$ the two-level kinetic hopping parameter, and $g \in \RR $  the coupling parameter between the photonic system and the two-level system. The operators
\begin{equation*}
    \qop =  \frac{1}{\sqrt{2\omega}} \qty( \hat{a}^\dagger +  \hat{a} )
    \quad \text{and} \quad
    \hat{p} = i \sqrt{\frac{\omega}{2}} \qty( \hat{a}^\dagger -  \hat{a} )
\end{equation*}
are the QHO position and momentum operators, respectively, and
$\hat{a}^\dagger$ ($\hat{a}$) the raising (lowering) QHO ladder operators. Moreover,
\begin{equation*}
    \hat\sigma_x = \begin{pmatrix} 0 & 1 \\ 1 & 0 \end{pmatrix}, \quad
    \hat\sigma_y = \begin{pmatrix} 0 & -i \\ i & 0 \end{pmatrix}, \quad
    \hat\sigma_z = \begin{pmatrix} 1 & 0 \\ 0 & -1 \end{pmatrix},
\end{equation*}
are the usual Pauli matrices.

In order to be able to do DFT, we will consider an extension of the quantum Rabi model which couples the photonic mode and the two-level system to external quantities $j \in \RR$ and $v \in \RR$, respectively.  We will refer to $\hat{H}_0$ as the \textit{internal} Hamiltonian, whereas the \textit{full} Hamiltonian is given by
\begin{equation}\label{eq:FullHamiltonian}
    \hat H(v,j) = \hat H_0 + v \hat \sigma_z + j \qop.
\end{equation}
The model setup of the full Hamiltonian is schematically illustrated in \cref{fig:QRabiSketch}. In this setting, applying a potential $v$ corresponds to tuning the level splitting, whereas applying a current $j$ displaces the QHO potential and induces an energy shift~\cite{Ruggenthaler2014}.

\begin{figure}[tbh]
    \centering
\usetikzlibrary{shapes.geometric, arrows}
\usetikzlibrary{3d}
\usetikzlibrary{decorations.pathmorphing}

\def\Ha{1.7}
\def\Hb{1.5}
\def\Wa{3}
\def\Wb{2}

\begin{tikzpicture}[decoration=snake]
        \coordinate (l1) at (-\Wa,-\Ha);
        \coordinate (l2) at (-\Wb,-\Ha);
        \coordinate (l3) at (-\Wb,\Ha);
        \coordinate (l4) at (-\Wa,\Ha);
        \coordinate (r1) at (\Wb,-\Ha);
        \coordinate (r2) at (\Wa,-\Ha);
        \coordinate (r3) at (\Wa,\Ha);
        \coordinate (r4) at (\Wb,\Ha);

        \draw[draw=none,fill=gray!70] (l1) -- (l2) to[out=130,in=230]  (l3) -- (l4) -- (l1);
        \draw[draw=none,fill=gray!70] (r1) -- (r2) -- (r3) -- (r4) to[out=310,in=50] (r1);

        \draw[thick,dashed] (-1.5,-0.7) -- (1.5,-0.7);
        \draw[thick,dashed] (-1.5,0.7) -- (1.5,0.7);
        \node at (1.8,-0.7) {$\psi_-$};
        \node at (1.8,0.7) {$\psi_+$};
        
        \draw (0,-0.7) circle (5pt);
        \draw (0,0.7) circle (5pt);
        \shadedraw (0,0) circle (5pt);

        \draw[decorate] (-2,0) -- (0,0);
        \node at (-1,0.3) {$\omega$};

        \draw[decorate,dotted,thick] (-4,0) -- (-2,0);
        \shadedraw[draw=none] (-4,0) circle (5pt);
        \node at (-4.4,0) {$j$};

        \node (v) at (1,0) {$v$};
        \draw[->,thick] (v) -- (1,0.6);
        \draw[->,thick] (v) -- (1,-0.6);

\end{tikzpicture}
    \caption{A schematic illustration of the full quantum Rabi model, \cref{eq:FullHamiltonian}, where $\omega$ is the photon frequency, the potential $v$ tunes the level splitting, and $j$ couples to the photonic degree-of-freedom. Here, $\psi_\pm$ are the components of the wavefunction in the eigenbasis of $\hat{\sigma}_z$.}
    \label{fig:QRabiSketch}
\end{figure}
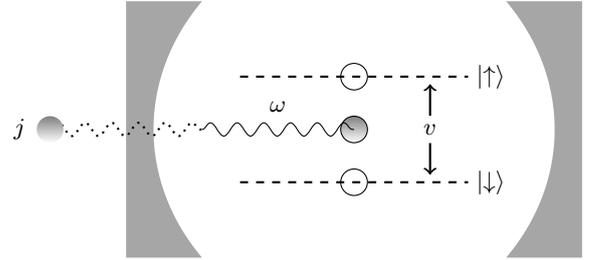

\subsection{Spaces and Domains}\label{subsec:Spaces}

The state space of the system is the Hilbert space $\mathcal{H} =  L^2(\RR, \CC) \otimes \CC^2  \simeq L^2(\RR, \CC^2)$, allowing us to represent a general state \(\psi \in \mathcal H\) as a two-component function with respect to the eigenbasis of the $\hat{\sigma}_z$ Pauli operator,
\begin{equation} \label{eq:pauli-representation}
    \psi(\qval) = \begin{pmatrix}
            \psi_+(\qval) \\
            \psi_-(\qval)
        \end{pmatrix},
    \quad
    \psi_+, \psi_- \in L^2(\RR , \CC ).
\end{equation}
We will work exclusively in this representation. The norm of \(\psi \in \mathcal H\) will be denoted by \(\norm{\psi}_\mathcal{H}\), and can be expressed in terms of the standard $L^2$-norms of \(\psi_\pm\) as $\norm{\psi}_\mathcal{H}^2 = \norm{\psi_+}_{L^2}^2 + \norm{\psi_-}_{L^2}^2$. The subscripts of the norms will be left implicit in the following. Likewise, the inner product between \(\psi, \varphi \in \mathcal H\) can be expressed by the inner product in \(L^2\) between the components, \(\braket{\psi}{\varphi} = \dua{\psi_+}{\varphi_+} + \dua{\psi_-}{\varphi_-}\). We will thus reserve the notation \(\braket{\cdot}{\cdot}\) for inner products between states in the full Hilbert space of the system, and use \(\dua{\cdot}{\cdot}\)  for inner products between the component wavefunctions in \(L^2\). The expectation value of an observable \(\hat A\) in the state \(\psi\) will be denoted \(\ev*{\hat A}{\psi}\). We will usually leave the \(\qval\)-dependence implicit in integrals, and unless stated otherwise, all integrals are definite integrals over $\RR$. 

All physical eigenfunctions are subject to the usual constraint of finite energy. For just the QHO, this means
\begin{equation*}
    \ev{\hat p^2 + \omega^2 \qop^2}{\psi} = \norm{\hat p\psi}^2 + \omega^2\norm{\qop\psi}^2 < \infty,
\end{equation*}
which is equivalent to $\norm{\hat p\psi}, \norm{\qop\psi} < \infty$ holding simultaneously. These conditions give us a space for the \emph{admissible} states, which is also the form domain~\cite{reed-simon-1} of the operator $\hat p^2 + \omega^2 \qop^2$,
\begin{equation*}
    Q_0 := \qty{ \psi \in \mathcal H : \norm{\hat p \psi},\norm{\qop\psi} < \infty }.
\end{equation*}
The Hamiltonians $\hat H_0$ and $\hat H(v,j)$ have a finite expectation value with respect to any $\psi\in Q_0$.

We will now define two variables, \(\sigma_\psi\) and \(\xi_\psi\), that will later become important as descriptors of the system in the context of DFT. In particular, let us define the polarization
\begin{equation*}
    \sigma_\psi
        = \ev{\hat \sigma_z}{\psi}
        = \norm{\psi_+}^2 - \norm{\psi_-}^2.
\end{equation*}
From the normalization condition \(\norm{\psi} = 1\) it follows that \(\abs{\sigma_\psi} \leq 1\) and that
\begin{equation} \label{eq:PauliProjectionConstraint}
     \norm{\psi_\pm}^2 = \frac{1 \pm \sigma_\psi}{2}.
\end{equation}
Similarly, we define the displacement of the photon field as
\begin{equation*}
    \xi_\psi
        = \ev{\qop}{\psi}
        = \int \qval \abs{\psi}^2 \, \dd{\qval},
\end{equation*}
using the shorthand \(\abs{\psi}^2 = \abs{\psi_+}^2 + \abs{\psi_-}^2\). We note that these variables satisfy $(\sigma_\psi, \xi_\psi) \in [-1, \, 1] \cross \RR$ and we sometimes call them, in analogy to standard DFT, a \emph{density pair}. They will be connected to the DFT treatment in \cref{subsec:variables}, but for now they serve only as notational shorthands.

\subsection{Properties of the Ground State}\label{subsec:GSProperies}
The point of departure for our investigation is the usual variational formulation for the ground-state energy corresponding to an external pair $(v,j) \in \RR^2$,
\begin{equation}
    \label{eq:GroundStateEnergy}
    E_0(v,j) := \inf_{\substack{\psi \in Q_0 \\ \|\psi\|=1}} \ev{\hat{H}(v,j)}{\psi}.
\end{equation}
Note that with the above definition alone, it is not guaranteed that there exists for all pairs \((v, j) \in \RR^2\) a normalized state in \(Q_0\) which realizes the ground-state energy \(E_0(v, j)\), i.e., that the infimum is in fact a minimum. Nevertheless, in this section we will first obtain a lower bound on the energy and then establish the existence of a ground state that is analytic, real, strictly positive and unique.

We begin by noting that for any $\psi \in Q_0$ the expectation value of the Hamiltonian is
\begin{align*}
    &\ev{\hat{H}}{\psi} = \frac{1}{2} \int \abs{\psi'}^2\dd{\qval} + \frac{\omega^2}{2} \int \qval^2 \abs{\psi}^2  \dd{\qval} \\
    &\qquad -2 t \Re \int \psi_+^* \psi_- \dd{\qval}+ g \int \qval \qty(\abs{\psi_+}^2 - \abs{\psi_-}^2) \dd{\qval} \\
    &\qquad + v \int \qty( \abs{\psi_+}^2 - \abs{\psi_-}^2) \dd{\qval} + j \int \qval \abs{\psi}^2 \dd{\qval}.
\end{align*}
By completing the squares, we have the simpler form 
\begin{equation}\label{eq:2HO}
    \begin{aligned}
        \ev{\hat{H}}{\psi} =& \sum_{\alpha \in \{ \pm \} } \int \qty[ \frac{1}{2} \abs{\psi_\alpha'}^2 + V_\alpha(\qval) \abs{\psi_\alpha}^2 ]  \dd{\qval} \\
        &- 2t \Re \int \psi_+^* \psi_- \dd{\qval}
    - \frac{j^2 + g^2}{2 \omega^2},
    \end{aligned}
\end{equation}
where the harmonic-oscillator potentials are 
\begin{equation} \label{eq:QHOPotentials}
    V_\pm(\qval) = \frac{\omega^2}{2} \qty(\qval + \frac{j \pm g}{\omega^2})^2  \pm \qty( v - \frac{j g}{\omega^2} )
\end{equation}
and where the kinetic term
\begin{equation*}
    - 2t \Re \int \psi_+^* \psi_- \dd{\qval} = -2t \Re \dua{\psi_+}{\psi_-} =\ev{-t\hat\sigma_x}{\psi}
\end{equation*}
did not change.
This shows that the quadratic form of $\hat{H}(v,j)$ can be expressed in terms of two coupled harmonic oscillators with potentials parametrically dependent on \(v\) and \(j\).  For the sake of brevity, let $C =(j^2 + g^2)/2 \omega^2$. Using the Cauchy--Schwarz inequality, we obtain from \cref{eq:PauliProjectionConstraint} that the kinetic term fulfills the estimate
\begin{align*}
    2 \Re \int \psi_+^* \psi_- \, \dd{\qval}
        &\le \Abs{2 \Re \int \psi_+^* \psi_- \, \dd{\qval}} \\
        \le 2 \Abs{\int \psi_+^* \psi_- \, \dd{\qval}}
        &\le 2 \norm{\psi_+} \norm{\psi_-}
        = \sqrt{1 - \sigma_\psi^2}.
\end{align*}
The Hamiltonian in \cref{eq:FullHamiltonian} thus has a well-defined ground-state energy, $E_0(v,j)$, for all external pairs $(v, j) \in \RR^2$, as the quadratic form of \(\hat H\) is bounded from below,
\begin{align*}
    &\ev{\hat{H}}{\psi} \\
        &\geq \sum_{\alpha \in \{ \pm \} }
            \int \qty[ \frac{1}{2} \abs{\psi_\alpha'}^2 + V_\alpha(\qval)\abs{\psi_\alpha}^2 ] \dd{\qval} 
         -t \sqrt{1 - \sigma_\psi^2} - C \\ 
    &\geq \sum_{\alpha \in \{ \pm \} } \frac{1 + \alpha \sigma_\psi}{2}
        \qty[\frac{\omega}{2} + \alpha \qty(v - \frac{j g}{\omega^2})] 
         -t \sqrt{1 - \sigma_\psi^2} - C \\
    &= \frac{\omega}{2} + \sigma_\psi\left( v - \frac{j g}{\omega^2}\right)
        - t \sqrt{1 - \sigma_\psi^2} - C \\
    &\geq \frac{\omega}{2}
        - \sqrt{t^2 + \qty(v - \frac{j g}{\omega^2})^2}
        - \frac{j^2 + g^2}{2 \omega^2}. \numberthis{\label{eq:lower-bound-ev-of-H}}
\end{align*}
Here, we have used \cref{eq:PauliProjectionConstraint} to attain the relative contributions of each shifted QHO to the energy, and the fact that the shifted QHO with potentials given by \cref{eq:QHOPotentials} independently have the energy solutions
\begin{equation*}
    E_\pm^n = \qty(n + \frac{1}{2}) \omega \pm \qty( v - \frac{j g}{\omega^2} )
    \quad \text{with} \quad n \in \mathbb{N}_0.
\end{equation*}
The final inequality can be shown by differentiating the preceding expression with respect to \(\sigma_\psi\), and setting the result to zero. The extremum is attained at
\begin{equation*}
    \sigma_\psi^* = - \frac{K}{\sqrt{t^2 + K^2}},
    \quad K = v - \frac{j g}{\omega^2}.
\end{equation*}
Note that \(\abs{\sigma_\psi^*} < 1\). Since the second derivative is positive for all \(\sigma_\psi \in (-1, 1)\) this extremum corresponds to a minimum, and it can easily be seen that this value is smaller than the value at the end points \(\sigma_\psi = \pm 1\).

Together with the boundedness below, we can now use the knowledge that for potentials $V_\pm(q)\to\infty$ as $|q|\to\infty$, as it is the case for the harmonic oscillator, the spectrum of the Hamiltonian is always discrete~\cite{simon2009} (this can also directly be argued from the compactness of the resolvent). This means that we always have a ground-state solution for any external pair \((v, j) \in \RR^2\).

Next, we show that all eigenstates of the time-independent Schrödinger equation with Hamiltonian from \cref{eq:FullHamiltonian} are analytic functions. Any eigenstate \(\psi\) with eigenvalue \(E\) is the solution to a system of 2nd order ODEs:
\begin{subequations}
    \begin{align}
        \psi_+''(\qval) &= p_+(\qval) \psi_+(\qval) - 2 t \psi_-(\qval), \label{eq:coupled-ODEs-a} \\
        \psi_-''(\qval) &= p_-(\qval) \psi_-(\qval) - 2 t \psi_+(\qval) \label{eq:coupled-ODEs-b}.
    \end{align}
\end{subequations}
The coefficients \(p_\pm(\qval)\) are the polynomials in $\qval$ given by
\begin{equation*}
  p_\pm(\qval) = \omega^2 \qval^2 + 2 (j \pm g) \qval - 2 (E \mp v).
\end{equation*}
Setting \(\phi_1 = \psi_+\), \(\phi_2 = \psi_-\), \(\phi_3 = \psi_+'\), and \(\phi_4 = \psi_-'\) lets us express these equations as a system of 1st order ODEs
\begin{equation*}
    \begin{pmatrix}
        \phi_1' \\ \phi_2' \\ \phi_3' \\ \phi_4'
    \end{pmatrix} =
    \begin{pmatrix}
        0 & 0 & 1 & 0 \\
        0 & 0 & 0 & 1 \\
        p_+(\qval) & -2t & 0 & 0 \\
        -2t & p_-(\qval) & 0 & 0
    \end{pmatrix}
    \begin{pmatrix}
        \phi_1 \\ \phi_2 \\ \phi_3 \\ \phi_4
    \end{pmatrix}.
\end{equation*}
This setting shows that every solution to the time-independent Schrödinger equation is \emph{analytic}, e.g.\ by \citet[Theorem~5.9 and Corollary~5.10]{coddington1997-book}.

For the real-valuedness of the components it is enough to see that the real and imaginary parts of $\psi_\pm(\qval)$ decouple in the expression of the quadratic form $\ev*{\hat H}{\psi}$, and that the minimization can be carried out for the real and imaginary parts separately. This decoupling is readily apparent from \cref{eq:2HO} for all terms except $-2t \Re \dua{\psi_+}{\psi_-}$. But also for this term we have
\begin{equation*}
    \begin{aligned}
        & \ev*{\hat \sigma_x}{\psi}
            = 2 \Re \dua{\psi_+}{\psi_-}
            = 1 - \|\psi_+ -\psi_-\|^2\\
            &= 1 - \|\Re\psi_+ -\Re\psi_-\|^2 - \|\Im\psi_+ -\Im\psi_-\|^2
    \end{aligned}
\end{equation*}
by the polarization identity, from which it is clear that the real and imaginary parts decouple for this term as well, meaning that \(\psi\) can be chosen to be real-valued.

Next, we demonstrate non-negativity. Let $\psi$ be any admissible state and define the level sets $P_\pm = \{ \qval \in \RR : \psi_\pm(\qval) \geq 0 \}$ and $M_\pm = \{ \qval \in \RR : \psi_\pm(\qval) < 0 \}$. We can then construct the non-negative state
\begin{equation*}
    \begin{pmatrix}
        \tilde{\psi}_+(\qval) \\ \tilde{\psi}_-(\qval)
    \end{pmatrix}
      = \left\{
        \begin{array}{ll}
          \qty(\psi_+(\qval),\, \psi_-(\qval))^\top,   & \qval \in P_+ \cap P_- \\
          \qty(\psi_+(\qval),\, -\psi_-(\qval))^\top,  & \qval \in P_+ \cap M_- \\
          \qty(-\psi_+(\qval),\, \psi_-(\qval))^\top,  & \qval \in M_+ \cap P_- \\
          \qty(-\psi_+(\qval),\, -\psi_-(\qval))^\top, & \qval \in M_+ \cap M_-.
        \end{array}\right.
\end{equation*}
By direct calculation, we see that the constraints and all terms in $\expval*{\hat{H}_0}{\tilde{\psi}}$ are unchanged by the transformation \(\psi \mapsto \tilde \psi\), except the kinetic hopping term. However, by the polarization identity again
\begin{align*}
    -t \expval*{\hat{\sigma}_x}{\tilde{\psi}}
        &= -2t \Re \dua{\tilde{\psi}_+}{\tilde{\psi}_-}
        = -t + t\norm{\tilde{\psi}_+ - \tilde{\psi}_-}^2,
\end{align*}
where the last term can be expanded as
\begin{align*}
    \norm{\tilde{\psi}_+ - \tilde{\psi}_-}^2
        &= \int\displaylimits_{\mathclap{P_+ \cap P_-}} \abs{\psi_+ - \psi_-}^2
          + \int\displaylimits_{\mathclap{M_+ \cap M_-}} \abs{\psi_+ - \psi_-}^2  \\
          &+ \int\displaylimits_{\mathclap{P_+ \cap M_-}} \abs{\psi_+ + \psi_-}^2
          + \int\displaylimits_{\mathclap{M_+ \cap P_-}} \abs{\psi_+ + \psi_-}^2.
\end{align*}
We then see that the two integrals on the first line remained unchanged by the transformation \(\psi \mapsto \tilde \psi\), whereas the two integrals on the second line did not increase as their integrands contain a sum of a positive and a negative part. It then follows that the transformation \(\psi \mapsto \tilde \psi\) does \textit{not} increase the energy and that $\tilde\psi$ must remain a ground state if $\psi$ is one. We may thus conclude that the components of the ground state can always be chosen to be non-negative.

Now, take such an analytic and non-negative ground state and assume it has $\psi_+(\qval)=0$ at some $\qval\in\RR$. Since it is non-negative it must hold $\psi_+'(\qval)=0$ and $\psi_+''(\qval)\geq 0$. Then by \cref{eq:coupled-ODEs-a}, it follows from $\psi_-(\qval)\geq 0$ that $\psi_+''(\qval)\leq 0$, so actually $\psi_+''(\qval)=0$. But by the same equation this means also $\psi_-(\qval)=0$ and thus from non-negativity that $\psi_-'(\qval)=0$ and from \cref{eq:coupled-ODEs-b} that $\psi_-''(\qval)=0$. But this just means that we have a zero solution since the initial values at $\qval$ are zero, so we arrived at a contradiction. The argument for $\psi_-(\qval)=0$ is equivalent, so we can conclude that $\psi(\qval)>0$ for all $\qval\in\RR$. Analyticity also implies that the expectation values $\sigma_\psi = \pm 1$ that would rely on $\psi_\mp=0$ cannot be achieved by a ground state, nor by any other eigenstate, since \cref{eq:coupled-ODEs-a,eq:coupled-ODEs-b} show that this would imply a zero wavefunction.

Note that we carefully state that the components of the ground state ``can be chosen'' real and strictly positive, since we could not yet rule out the existence of other ground states that do not have this feature. However, the ground state of the quantum Rabi model must \emph{always} be strictly positive by an argument that uses that the imaginary-time evolution is positivity improving~\cite{Hirokawa2014}. Now, since a second ground state, $\psi_2$, must be orthogonal to the first one, $\psi_1$, but is positive as well, so we inevitably get $\braket{\psi_1}{\psi_2}>0$ which is a contradiction. We thus conclude that the ground state is also unique.

\subsection{Results from the Hypervirial Theorem}\label{subsec:HyperVirial}
With the system being fully defined by its Hamiltonian, \cref{eq:FullHamiltonian}, and having established some important properties of its ground state, we now move to deriving some useful relations between expectation values of operators with respect to any eigenstate (in particular the ground state). The hypervirial theorem offers a convenient tool for that task. The hypervirial theorem~\cite{hirschfelder1960classical} is the simple result that for any time-independent operator $\hat A$ and an eigenstate $\psi$ of $\hat H$ it holds by the Ehrenfest theorem (with $\time$ the time variable)
\begin{equation*}
    \dtime\ev{\hat A}{\Psi} = i\ev{[\hat H,\hat A]}{\Psi} = 0.
\end{equation*}
This is particularly useful if $\hat A$ is chosen such that $[\hat H,\hat A]$ yields an operator of interest. This method was already applied to the Pauli--Fierz Hamiltonian to get field-mode virial theorems and expressions for the coupling energy~\cite{theophilou2020virial}. In the following, we will use the Heisenberg picture to give the time-derivative (first and sometimes second order) for a variety of operators that then give useful relations in the stationary setting. The calculation just requires the canonical commutator relations between $\qop$ and $\hat p$ and the commutator relations between the Pauli matrices. For example, choosing $\hat{A} = \qop$ gives 
\begin{equation}\label{eq:dt-x}
    \dtime\qop = \hat p, \quad \ddtime\qop=\dtime\hat p=-\omega^2\qop-g\hat\sigma_z-j .
\end{equation}
For the expectation value, always with respect to an eigenstate $\psi$ of $\hat H$, we use the shorthand notation $\ev*{\hat O}=\ev*{\hat O}{\psi}$ in this section. Then the expectation values of the expressions above are
\begin{align}
    &\ev{\hat p} = 0, \nonumber\\
    &\omega^2\ev{\qop} + g\ev{\hat\sigma_z} + j = 0.
    \label{eq:sigma-xi-j-relation}
\end{align}
Here, the last equation is the analogue of Maxwell's equation in this reduced setting and connects $\ev{\qop}$ and $\ev{\hat\sigma_z}$ with the external $j$. This important relation will be used later multiple times. If we add the second order for $\hat p$,
\begin{equation*}
    \quad \ddtime\hat p = -\omega^2\hat p+2gt\hat\sigma_y,
\end{equation*}
then since we already know $\ev{\hat p}=0$, we get
\begin{equation}\label{eq:ev-sigmay}
    \ev{\hat\sigma_y} = 0.
\end{equation}
This relation gets obvious if one writes it out in wavefunction components and takes into account that the ground-state wavefunction is always real (\Cref{subsec:GSProperies}), yet we now additionally established that it holds for any eigenstate,
\begin{equation*}
    \int \left( \psi_+^* \psi_- - \psi_-^* \psi_+ \right) = 0.
\end{equation*}

Another interesting choice is $\hat A = \qop\hat p$,
\begin{equation*}
    \dtime \qop\hat p = \hat p^2-\omega^2\qop^2-g\hat\sigma_z\qop-j\qop,
\end{equation*}
which leads to the usual virial theorem for the photon mode that relates to the coupling energy and the energy from current $j$,
\begin{equation}\label{eq:classical-VR}
    \ev{\hat p^2} - \omega^2\ev{\qop^2} = g\ev{\hat\sigma_z\qop} + j\ev{\qop}.
\end{equation}
Next, the Pauli operator $\hat\sigma_z$ yields the polarization as an expectation value,
\begin{align*}
    &\dtime \hat\sigma_z = -2t\hat\sigma_y,\\
    &\ddtime\hat\sigma_z = -2t\dtime\hat\sigma_y = -4t^2\hat\sigma_z - 4tg\hat\sigma_x\qop - 4tv\hat\sigma_x.
\end{align*}
The expectation value of the first equation is just \cref{eq:ev-sigmay} again. However, the second equation allows us to determine the external potential $v$ from expectation values via 
\begin{equation}\label{eq:force-balance}
    t\ev{\hat\sigma_z} + g\ev{\hat\sigma_x\qop} + v\ev{\hat\sigma_x} = 0,
\end{equation}
and will thus be important for later applications. This expression is analogous to the force-balance equation that can also be employed in standard DFT and QEDFT to derive an exchange-correlation potential~\cite{tancogne2024exchange,lu2024electron}. This procedure will be showcased in \cref{subsec:eff-pot} for the quantum Rabi model. Moreover, we give one last result from the hypervirial theorem of particular interest since it includes the coupling term,
\begin{align}
    &\dtime \hat\sigma_z\hat p = -\omega^2\hat\sigma_z\qop -2t\hat\sigma_y\hat p -g-j\hat\sigma_z,\nonumber\\
    &\omega^2\ev{\hat\sigma_z\qop}+2t\ev{\hat\sigma_y\hat p} +g+j\ev{\hat\sigma_z}=0. \label{eq:virial-sigmaz-p}
\end{align}
Of course, many more such exact relations can be derived, but the above are the most significant for this study of the quantum Rabi model.

Before moving on to the analysis of the quantum Rabi model from a mathematical DFT perspective, let us briefly venture into its immediate extension. The Reader should be forewarned that this section and \cref{subsec:HKDickeModel} are somewhat peripheral to the main objectives of this article. However, it does provide some useful insights to notions not present in the quantum Rabi model, including a characterization of the set of densities where the Hohenberg--Kohn theorems hold.

\subsection{The Dicke Model}\label{subsec:DickeModel}
Recall from our discussion of the model hierarchy in \cref{subsec:QEDModels}, see in particular \cref{fig:ModelHierarchy}, that the immediate extensions of the quantum Rabi model are the Dicke model and its extension, the multi-mode Dicke model. For the sake of simplicity, let us only look at the Dicke model. However, the inclusion of multiple photonic modes does not significantly complicate the problem, and the interested Reader is referred to \citet{Bakkestuen2024}.

The Dicke model describes a  system of $N$ two-level systems coupled to a single quantum harmonic oscillator. By extension of the state space for the quantum Rabi model, the state space of the Dicke model is the Hilbert space $\mathcal{H}_\text{Dicke} = L^2(\RR , \CC ) \otimes \CC ^{2^N}\simeq L^2(\RR , \CC ^{2^N})$. Thus, a state $\psi \in  L^2(\RR , \CC ^{2^N})$ has 4 components in the position basis if $N=2$ and 8 components if $N=3$. In analogy to \cref{eq:pauli-representation}, for $N=2$ we write
\begin{equation*}
    \psi(\qval) = \begin{pmatrix} \psi_{++}(\qval) \\ \psi_{+-}(\qval) \\ \psi_{-+}(\qval) \\ \psi_{--}(\qval) \end{pmatrix},
    \quad \begin{matrix}
        \psi_{\alpha\beta} \in L^2(\RR , \CC ) \\ \forall\, \alpha,\beta = \pm.
    \end{matrix}
\end{equation*}
The appropriate inner product on $L^2(\RR , \CC ^{2^N})$, denoted in the same manner as for the quantum Rabi model, is then 
\begin{equation*}
    \braket{\varphi}{\psi} = \sum_{\alpha_1=\pm }\cdots\sum_{\alpha_N =\pm} \langle \varphi_{\alpha_1,\dots,\alpha_N}, \psi_{\alpha_1,\dots,\alpha_N}  \rangle.
\end{equation*}
The displacement of a state is then 
\begin{equation*}
    \xi_\psi := \ev{\qop}{\psi} = \sum_{\alpha_1=\pm }\cdots\sum_{\alpha_N =\pm} \int \qval \abs{\psi_{\alpha_1,\dots,\alpha_N}(\qval)}^2 \dd{\qval}.
\end{equation*}

In order to define the polarization variable, we need to extend the definition of the Pauli operators to act on the Hilbert space of the Dicke model. For any $1 \leq j \leq N$, the \textit{lifted Pauli operators} are defined as 
\begin{equation*}
    \hat{\sigma}_a^j = \id \otimes \cdots \otimes \id \otimes \underbrace{\hat{\sigma}_a}_{j\text{th}} \otimes \id \otimes \cdots \otimes \id,
\end{equation*}
where $a=x,y,z$. By employing the usual matrix forms, the lifted Pauli-$z$ matrices for $N=2$ are 
\begin{align*}
    \hat\sigma_z^1 = \begin{pmatrix}  1 &&& \\ &1&&\\&&-1& \\&&&-1 \end{pmatrix}, \quad \hat\sigma_z^2 = \begin{pmatrix}  1 &&& \\ &-1&&\\&&1& \\&&&-1 \end{pmatrix}.
\end{align*}
All the lifted Pauli matrices can then be collected into the vector $\hat\vsigma_a = (\hat\sigma_a^1,\dots,\hat\sigma_a^N)^\top \in (\CC ^{2^N\times 2^N})^N$. Consequently, the polarization vector is given by
\begin{equation*}
    \vsigma_\psi := \ev{\hat{\vsigma}_z}{\psi} = \qty(\sigma_1, \dots, \sigma_N)^\top, \quad \sigma_j = \ev{\hat{\sigma}_z^j}{\psi}. 
\end{equation*}
For $N=2$, along with the normalization constraint, this yields
\begin{align*}
    1 &= \norm{\psi_{++}}^2 + \norm{\psi_{+-}}^2 + \norm{\psi_{-+}}^2 + \norm{\psi_{--}}^2, \\ 
    \sigma_1 &= \norm{\psi_{++}}^2 + \norm{\psi_{+-}}^2 - \norm{\psi_{-+}}^2 - \norm{\psi_{--}}^2, \\ 
    \sigma_2 &= \norm{\psi_{++}}^2 - \norm{\psi_{+-}}^2 + \norm{\psi_{-+}}^2 - \norm{\psi_{--}}^2,
\end{align*}
or alternatively 
\begin{align*}
    \frac{1 \pm \sigma_1}{2} &= \norm{\psi_{\pm+}}^2 + \norm{\psi_{\pm-}}^2, \\
    \frac{1 \pm \sigma_2}{2} &= \norm{\psi_{+\pm}}^2 + \norm{\psi_{-\pm}}^2.
\end{align*}
Equipped with the vector of the lifted Pauli operators, the internal \textit{Dicke-model Hamiltonian} is 
\begin{equation}\label{eq:DickeHamiltonian}
    \hat{H}_0 = \frac{1}{2} \hat{p}^2 + \frac{\omega^2}{2} \qop^2 - \mathbf{t}\cdot \hat{\vsigma}_x + \mathbf{g} \cdot \hat{\vsigma}_z \qop,
\end{equation}
where $\mathbf{t}\in \RR ^N$ is the kinetic hopping parameters and $\mathbf{g}\in \RR ^N$ the coupling parameters. For simplicity, the form domain of the Dicke Hamiltonian will also be denoted $Q_0$. Coupled to the external quantities $\vv \in \RR ^N$ and $j\in\RR$, the full Hamiltonian then is  
\begin{equation*}
    \hat{H} = \hat{H}_0 + \vv \cdot \hat{\vsigma}_z + j\qop.
\end{equation*}
Despite having introduced the setting of the Dicke model, the following results, apart from \cref{subsec:HKDickeModel} on the Hohenberg--Kohn theorem, will be mainly for the quantum Rabi model only.

\section{Hohenberg--Kohn Theorems}\label{sec:HK}

\subsection{Internal Variables for the Quantum Rabi Model}\label{subsec:variables}
A common starting point for developing a DFT is to establish the connection between what are known as the internal and external variables of the system. As briefly mentioned in \cref{sec:DFTBackground}, the famous Hohenberg--Kohn theorem establishes that the internal variable---the electronic density---determines the external variable---the external potential---uniquely up to an additive constant. This mapping from the internal to the external variables firmly establishes the internal variable as a descriptor of the system (for a fixed number of electrons) in the case of Coulombic interactions. 

Our point of departure to develop a QEDFT using the quantum Rabi model, and similarly for the Dicke model, will be to establish a Hohenberg--Kohn theorem. Recall from \cref{subsec:Spaces} that the full Hamiltonian is described by the pair of external quantities $(v,j)$. Furthermore, we introduced two expectation values of particular interest, 
\begin{equation*}
	\sigma_\psi = \ev{\hat{\sigma}_z}{\psi} \quad \text{and} \quad \xi_\psi = \ev{\qop}{\psi}, 
\end{equation*}
the polarization and the displacement of the photonic field, respectively. Moreover, in treating $\sigma$ and $\xi$ as free variables, not necessarily with a reference to a state, the previously given constraints demand $\sigma \in [-1, 1]$ and $\xi \in \mathbb{R}$.  The goal of this section is then to derive a Hohenberg--Kohn-type result, mapping a pair $(\sigma,\xi) \in [-1,1] \times \mathbb{R}$ to a external pair $(v,j)\in \mathbb{R}^2$ that yields exactly this density pair in the ground state. Deriving such a Hohenberg--Kohn-type result establishes the pair $(\sigma,\xi)$ as the internal (density) variables of the system that fully determine the Hamiltonian at hand and thus also the ground state(s). This then sets the stage for developing a QEDFT using a Levy--Lieb functional in \cref{sec:LL}.

Firstly, we show this for the quantum Rabi model. The proof strategy that we employ consists of separating the mapping from density pairs back to external pairs into two parts: First (HK1), the mapping from densities to ground states and second (HK2), the map from such states to external quantities. The reason is that while the first part is almost trivial and holds for all kinds of different DFTs, the second one is highly dependent on the system under consideration. This division of the Hohenberg--Kohn theorem is further detailed in \citet{penz2023structure1}. Later, a Hohenberg--Kohn theorem for the Dicke model is also proven and we illustrate the notion of ``regular'' densities. Note also that the proof given for \cref{thrn:HKDicke} is somewhat different from the proof given in \citet{Bakkestuen2024}, leading to an alternative characterization for regular densities. 

\subsection{Hohenberg--Kohn Theorem for the Quantum Rabi Model}
\label{subsec:HK-QRabi}

The first part, also sometimes called the \emph{weak Hohenberg--Kohn theorem}~\cite{tellgren2018uniform}, shows that if two states have the same internal variables \((\sigma, \xi)\) and they are ground states of Hamiltonians which differ only by the values of the external pair \( (v,j) \), then both states will also be ground states of the other Hamiltonian. We here use the notation $\psi\mapsto(\sigma,\xi)$ to state that a normalized $\psi$ has $\sigma_\psi=\sigma$ and $\xi_\psi=\xi$.

\begin{lemma}[Weak Hohenberg--Kohn, HK1] \label{lemma:WeakHK}
    Suppose $\psi_1,\psi_2\in Q_0$ are ground states of
    $\hat{H}(v_1, j_1)$ and $\hat{H}(v_2, j_2)$, respectively. If both
    $\psi_1, \psi_2 \mapsto (\sigma, \xi)$, then $\psi_2$ is a
    ground state of $\hat{H}(v_1,j_1)$ and $\psi_1$ is a ground state of
    $\hat{H}(v_2,j_2)$.
\end{lemma}
In order to show this, first note that \(\sigma_{\psi_1} = \sigma_{\psi_2} = \sigma\) and \(\xi_{\psi_1} = \xi_{\psi_2} = \xi\) by assumption. Then it follows from \cref{eq:GroundStateEnergy} that any ground state $\psi_i$ of $\hat{H}(v_i,j_i)$ is a minimizer of
\begin{equation*}
    E_0(v_i, j_i) = \inf_{ \psi \mapsto (\sigma,\xi)}
        \ev{\hat{H}_0}{\psi} + \sigma v_i+ \xi j_i,\quad i=1,2.
\end{equation*}
But $i=1$ and $i=2$ include exactly the same variational problem, such that any such minimizer $\psi_i$ is a ground state for both, $\hat{H}(v_1,j_1)$ and $\hat{H}(v_2,j_2)$.

Note that this type of argument works for any variant of DFT, as long as the external quantities couple only to the internal density variables~\cite{penz2023structure1}. The second part of the Hohenberg--Kohn theorem, referred to as HK2, is then more geared to the special structure of the problem.

\begin{lemma}[HK2]
    \label{lemma:HK2}
    If two Hamiltonians $\hat{H}(v_1,j_1)$ and $\hat{H}(v_2,j_2)$ share any eigenstate then $v_1 = v_2$ and $j_1=j_2$.
\end{lemma}
To prove this, let the two Hamiltonians \(\hat H(v_1, j_1)\) and \(\hat H(v_2, j_2)\) share the common eigenstate \( \psi \in Q_0\) and denote the respective eigenvalues by \(E_1\) and \(E_2\). Then, by the Schrödinger equation,
\begin{equation*}
    [ (v_1 - v_2) \hat \sigma_z + (j_1 - j_2) \qop ]  \psi
        = (E_1 - E_2)  \psi.
\end{equation*}
By assuming $j_1 \neq j_2$, one obtains that
\begin{equation*}
    \qop \psi_\pm
        = \qty[ \frac{E_1 - E_2}{j_1 - j_2} \mp \frac{v_1 - v_2}{j_1 - j_2}] \psi_\pm.
\end{equation*}
However, the operator \(\qop\) has no square-integrable eigenfunctions apart from the trivial solution $\psi_\pm \equiv 0$. Consequently, $j_1 = j_2$, since we of course cannot have \( \psi \) identically zero. This in turn implies that
\begin{equation*}
    \pm (v_1 - v_2) \psi_\pm = (E_1 - E_2) \psi_\pm.
\end{equation*}
We can multiply these two equations by $\psi_\pm^*$ and integrate over $\qval$ such that
\begin{equation*}
    \pm (v_1 - v_2) \|\psi_\pm\|^2 = (E_1 - E_2) \|\psi_\pm\|^2.
\end{equation*}
In this equation, we can also drop the squares on both sides and get
\begin{equation}\label{eq:HK-v1-v2}
    \pm (v_1 - v_2) \|\psi_\pm\| = (E_1 - E_2) \|\psi_\pm\|.
\end{equation}
(While this does not change anything here, we will need this form for the arguments of the next section.) Now, by what we know about eigenstates from \cref{subsec:GSProperies}, both components have $\|\psi_\pm\| \neq 0$ and the equations above give $v_1 - v_2 = E_1 - E_2 = v_2 - v_1$. So we must have $v_1 = v_2$, in which case also \(E_1 = E_2\) and the Hamiltonians are identical.

A consequence of \cref{lemma:HK2} is that the mapping from \((v, j)\) to any eigenstate \(\psi\) of \(\hat H (v, j)\) is injective.  Since we already know that an eigenstate has always $|\sigma_\psi|<1$, we will remove the $\sigma=\pm 1$ from the following statement for the complete mapping $(v,j) \mapsto (\sigma,\xi)$. We will, as defined in \cref{subsec:HKDickeModel} also for the more general case of the Dicke model, call the $\sigma=\pm 1$ critical, while $\sigma\in(-1,1)$ are regular. If we combine the two results from \cref{lemma:WeakHK} and \cref{lemma:HK2} we get a full Hohenberg--Kohn theorem for the quantum Rabi model.

\begin{theorem}[Hohenberg--Kohn for the Quantum Rabi Model]\label{thrm:HKRabi}
    Any $(\sigma,\xi)\in(-1,1)\times\RR$ that is the density pair of a ground state uniquely determines an external pair $(v,j)\in \RR ^2$. That is, the mapping
    \[\RR ^2 \ni (v,j) \longmapsto (\sigma,\xi) \in (-1,1)\times \RR \]
    is an injection.
\end{theorem}

Note that unlike the original Hohenberg--Kohn theorem~\cite{Hohenberg1964}, \cref{thrm:HKRabi} uniquely determines the external quantities, i.e., not only up to an additive constant. Then as noted in \cref{subsec:variables}, \cref{thrm:HKRabi} establishes the pair $(\sigma,\xi)$ as the internal (density) variables of the system. The surjectivity of this mapping will be the result of the $v$-representability (actually $(v,j)$-representability, but we will stick to the usual DFT terminology here) of all regular density pairs $(\sigma,\xi) \in (-1,1)\times \RR$ and will be demonstrated in \cref{sec:optimizers-gs}.

\subsection{Hohenberg--Kohn Theorem for the Dicke Model}
\label{subsec:HKDickeModel}

In generalizing from the quantum Rabi model to the Dicke model, the Hohenberg--Kohn theorem is one of the complicating factors. To achieve it, as we will see in this section, we have to introduce the useful concept of a \textit{regular} polarization vector $\vsigma$. But first let us briefly state the analogous weak Hohenberg--Kohn result for the Dicke model.

\begin{lemma}[Weak Hohenberg--Kohn, HK1]\label{lemma:HK1Dicke}
    Suppose that $\psi_1,\psi_2\in Q_0$ are ground states of $\hat{H}(\vv _1,j_1)$ and $\hat{H}(\vv _2,j_2)$ respectively, where $(\vv _1,j_1)\in \RR^N \times \RR$ and $(\vv _2,j_2) \in \RR^N\times\RR$ are two external pairs. If both $\psi_1,\psi_2\mapsto(\vsigma,\xi)$, then $\psi_2$ is a ground state of $\hat{H}(\vv _1,j_1)$ and $\psi_1$ is a ground state of $\hat{H}(\vv _2,j_2)$
\end{lemma}

The proof of this result is exactly the same as the proof of \cref{lemma:WeakHK}, and is thus a prime example of a case where the generalization from the quantum Rabi model to the Dicke model does not complicate matters. However, this is not entirely true for the generalization of \cref{lemma:HK2}, which requires the following definition of a regular polarization vector. Note that \citet{Bakkestuen2024} employ a different definition of regular $\vsigma$ that is adapted to the proof technique used there. However, these two definitions are equivalent.

\begin{definition}[Regular Polarization]\label{def:RegularDensity-Dicke-new}
    A $\vsigma \in [-1,1]^N$ is called \emph{regular} if for every $\chi \in \RR_+^{2^N}$ that has $\|\chi\|=1$ and $\ev{\hat\vsigma_z}{\chi}=\vsigma$ one has $\{\chi, \hat\sigma_z^1\chi, \ldots, \hat\sigma_z^N\chi\}$ as a set of linear independent vectors. The set of all regular $\vsigma$ is denoted $\mathcal{R}_N$. Any $\vsigma \in [-1,1]^N\setminus \mathcal{R}_N$ is not regular, and is called \emph{critical}. 
\end{definition}

The relevance of this definition becomes immediately clear if we try to prove the second part of the Hohenberg--Kohn theorem for $N\geq 2$ with the same method as in \cref{lemma:HK2} before. We introduce $\chi\in\RR_+^{2^N}$, with components $\chi_{\alpha_1,\dots,\alpha_N} = \|\psi_{\alpha_1,\dots,\alpha_N}\|$. Then $\ev{\hat\vsigma_z}{\chi}=\ev{\hat\vsigma_z}{\psi}=\vsigma$ and since the steps up to \cref{eq:HK-v1-v2} are completely analogously we have
\begin{equation*}
    (\vv _1-\vv _2)\cdot \hat{\vsigma}_z \chi = (E_1 - E_2) \chi.
\end{equation*}
Here, each $\hat{\sigma}_z^j$ in $\hat{\vsigma}_z$ is purely diagonal and just applies the correct sign for every component of $\chi$.
Now assume that $\vsigma$ is a regular polarization vector, then all components $\chi, \hat\sigma_z^1\chi, \ldots, \hat\sigma_z^N\chi$ in the equation above are linearly independent and it directly follows $\vv _1-\vv _2=0$ and $E_1 - E_2=0$ as required. This proves the full Hohenberg--Kohn theorem for all regular $\vsigma\in\mathcal{R}_N$.

\begin{theorem}[Hohenberg--Kohn for the $N$-site Dicke Model]\label{thrn:HKDicke}
    Any $(\vsigma,\xi)\in \mathcal{R}_N\times \RR $ that is the density pair of a ground state uniquely determines an external pair $(\vv ,j)\in \RR ^{N+1}$. That is, the mapping 
    \begin{equation*}
        \RR^{N+1} \ni (\vv ,j) \longmapsto (\vsigma,\xi) \in \mathcal{R}_N\times \RR 
    \end{equation*}
    is an injection.
\end{theorem}

To exemplify the notion of regular polarizations, we will now show in three examples, for $N=1,2,3$, how the set of critical polarization vectors looks like.

\begin{example}[$N=1$]
    This is the case of the quantum Rabi model. Having $\sigma\in[-1,1]$ regular demands that all $\chi\in\RR_+^2$ with $\|\chi_+\|^2-\|\chi_-\|^2=\sigma$ have $\{\chi,\hat\sigma_z\chi\}$ linear independent. But $(\chi_+,\chi_-)$ and $(\chi_+,-\chi_-)$ are easily seen to be linear independent whenever both $\chi_\pm\neq 0$. This means the critical $\sigma$ are those that come from $\chi_\pm=0$, exactly $\sigma=\pm 1$. The regular set is then $\mathcal{R}_1=(-1,1)$.
\end{example}

\begin{example}[$N=2$]
    Now $\chi\in\RR_+^4$ and we need to have the three vectors $\{\chi,\hat\sigma_z^1\chi,\hat\sigma_z^2\chi\}$ linear independent. Instead of writing things out in components, we will argue combinatorially. Three out of four non-zero components are enough to have the three vectors linearly independent, consequently we get an critical $\vsigma$ from $\chi$ whenever two or three components vanish. Three vanishing components just leaves four possibilities, $\chi_1=(1,0,0,0)^\top$ and its permutations $\chi_2,\chi_3,\chi_4$, which map to $\vsigma=[\pm 1,\pm 1]$, precisely the four corners of the polarization square $[-1,1]^2$. A linear combination of two $\chi_i$ still has two zero-components and maps to the straight lines that connect two corners. We thus have $\binom{4}{2} = 6$ such lines that represent critical polarizations. The situation is depicted in \cref{fig:SigmaSet2D}.
\end{example}

\begin{figure}
    \centering
\usetikzlibrary{shapes.geometric, arrows}
\usetikzlibrary{3d}

\tdplotsetmaincoords{0}{0}
\def\x{-2.3}
\def\y{2}
\def\z{0}
\def\r{2}
\begin{tikzpicture}[scale=2,tdplot_main_coords]
\tikzstyle{grid}=[thin,color=red,tdplot_rotated_coords]  
  \draw [thick] (0,0,0) -- (2,0,0);
  \draw [thick] (0,0,0) -- (0,2,0);
  \draw [thick] (2,2,0) -- (2,0,0);
  \draw [thick] (2,2,0) -- (0,2,0);
  
  \draw [dashed] (0,0,0) -- (2,2,0);
  \draw [dashed] (0,2,0) -- (2,0,0);
  
  
  \coordinate (a) at (\x,\y,\z);
  \coordinate (b) at (\x+\r,\y,\z);
  \coordinate (c) at (\x+\r/2,\y-\r/2,\z);
  
  \coordinate (d) at (\x,\r*0.33,\z);
  \coordinate (e) at (\x+\r,\r*0.33,\z);
  
  

\end{tikzpicture}
    \caption{The set of regular densities for $N=2$, $\mathcal{R}_2 \subset (-1,1)^2$, is the union of four congruent open triangles.}
    \label{fig:SigmaSet2D}
\end{figure}
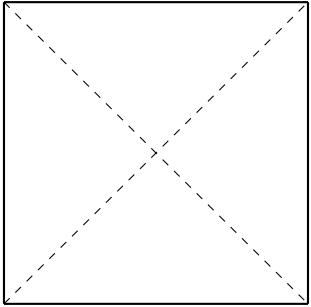

\begin{example}[$N=3$]
    We have $\chi\in\RR_+^8$ and argue as before. A critical $\vsigma$ comes from a $\chi$ that has between five and seven zero-components. On the other hand, four non-zero components are enough to have $\{\chi,\hat\sigma_z^1\chi,\hat\sigma_z^2\chi,\hat\sigma_z^3\chi\}$ linear independent. If $\chi$ has just one non-zero component, it maps to the 8 corners of the polarization cube $[-1,1]^3$. The $\binom{8}{2}=28$ ways how to combine these corners then come from $\chi$ with two non-zero components. These are the $12$ edges of the cube, plus $12$ diagonals on the faces, plus $4$ diagonals in the inside of the cube. Finally, we always take three of the corners of the cube and combine them into $\binom{8}{3}=56$ planes that belong to $\chi$ with three non-zero components. Here one sees that many of these combinations actually yield the same plane, so that only 20 remain: 6 faces, 6 from the diagonals on the faces connecting to the opposite face, and 8 that are formed by three face-diagonals.
    The situation is depicted in \cref{fig:SigmaSet3D}.
\end{example}

\begin{figure}
    \centering
\usetikzlibrary{shapes.geometric, arrows}
\tdplotsetmaincoords{80}{160}
\usetikzlibrary{3d}

\def\x{3.7}
\def\y{0}
\def\z{0}

\def\X{3.7}
\def\Y{0}
\def\Z{-1}

\def\r{2}
\def\h{1}
\begin{tikzpicture}[scale=2,tdplot_main_coords]
\tikzstyle{grid}=[thin,color=red,tdplot_rotated_coords]  
  \draw [thick] (0,0,0) -- (2,0,0);
  \draw [thick] (0,0,0) -- (0,2,0);
  \draw [thick] (2,2,0) -- (2,0,0);
  \draw [thick] (2,2,0) -- (0,2,0);
  
  \draw [thick] (0,0,2) -- (2,0,2);
  \draw [thick] (0,0,2) -- (0,2,2);
  \draw [thick] (2,2,2) -- (2,0,2);
  \draw [thick] (2,2,2) -- (0,2,2);
  
  \draw [thick] (0,0,2) -- (0,0,0);
  \draw [thick] (2,0,2) -- (2,0,0);
  \draw [thick] (0,2,2) -- (0,2,0);
  \draw [thick] (2,2,2) -- (2,2,0);

  \draw [dashed] (0,0,0) -- (2,2,2);
  \draw [dashed] (2,0,0) -- (0,2,2);
  \draw [dashed] (0,2,0) -- (2,0,2);
  \draw [dashed] (0,0,2) -- (2,2,0);
  
  \draw [dashed] (0,0,0) -- (2,0,2);
  \draw [dashed] (2,0,0) -- (0,0,2);
  
  \draw [dashed] (0,0,0) -- (0,2,2);
  \draw [dashed] (0,2,0) -- (0,0,2);
  
  \draw [dashed] (0,0,0) -- (2,2,0);
  \draw [dashed] (0,2,0) -- (2,0,0);
  
  \draw (2,0,0) -- (2,2,2);
  \draw (2,2,0) -- (2,0,2);
  
  \draw (0,2,0) -- (2,2,2);
  \draw (2,2,0) -- (0,2,2);
  
  \draw (0,0,2) -- (2,2,2);
  \draw (2,0,2) -- (0,2,2);
  
  \draw [dashed] (0,1,1) -- (2,1,1);
  \draw [dashed] (1,0,1) -- (1,2,1);
  \draw [dashed] (1,1,0) -- (1,1,2);

\end{tikzpicture}
    \caption{The set of regular polarizations for $N=3$, $\mathcal{R}_3 \subset (-1,1)^3$, is the union of open polyhedra created by 14 planes cutting through the open cube $(-1,1)^3$. These are the 6 planes extending from the diagonals on the faces plus the 8 planes formed by taking three face-diagonals.}
    \label{fig:SigmaSet3D}
\end{figure}
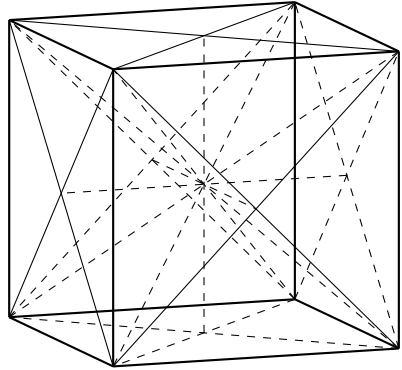

Already from the above examples up to $N=3$ one realizes that for general $N$, the regular set $\mathcal{R}_N$ is the union of disjoint open convex polytopes~\cite{Bakkestuen2024}. In $N=2$ the polytopes are simply triangles, whilst in $N=3$ they are polyhedra.

In summary, given a ground-state density pair $(\vsigma,\xi)$, the uniqueness for external pairs can only be guaranteed for regular $\vsigma$. Nevertheless, we never encountered a non-unique external pair in case of the Dicke model, in contrast to lattice models, where explicit counterexamples to a full Hohenberg--Kohn theorem are known~\cite{penz2021-Graph-DFT}. Note especially that such counterexamples always rely on ground-state degeneracy~\cite{penz2023geometry}, but such a degeneracy was also never observed in our numerical investigations of the Dicke model up to now.

\section{The Levy--Lieb Functional}\label{sec:LL}
\subsection{Definition}

From the investigation of the Hohenberg--Kohn theorems in the preceding section, the internal variables $\sigma \in [-1,1]$ and $\xi \in \RR$ arise as descriptors of the system. This justifies, in analogy to standard DFT, referring to them collectively as a density pair. Consider now the minimization of the internal energy $\psi \mapsto \ev{\hat{H}_0}{\psi}$ under the constraint of normalizable and admissible $\psi$ mapping to the correct density pair, i.e., $\sigma_\psi = \sigma$ and $\xi_\psi = \xi$. This problem naturally gives rise to what is commonly called the \textit{pure-state constrained-search functional}, also known as the \textit{Levy--Lieb functional}, introduced in \cref{sec:DFTBackground} for standard DFT. However, before defining this functional for the quantum Rabi model, let us turn to the question of whether a given density pair $(\sigma,\xi) \in [-1,1]\times \RR$ can even be represented by a state $\psi$.

\begin{theorem}($N$-representability) \label{thrm:N-Rep}
    For every density pair $(\sigma,\xi) \in [-1,1] \cross \RR$ there exists $\psi \in Q_0$ such that $\norm{\psi} = 1$, $\ev{\hat{\sigma}_z}{\psi} = \sigma$, and $\ev{\qop}{\psi} = \xi$.
\end{theorem}
Note here, that ``$N$-representability'' is a technical term from the standard DFT terminology and just means that the given constraints can be fulfilled by an $N$-particle state. In the quantum Rabi model, naturally, this $N$ has no further meaning. However, \cref{thrm:N-Rep} also holds for the Dicke model, in which case the $N$ is not to be confused with the number of two-level systems.  We give a simple and constructive proof.

\begin{proof}[Proof of \cref{thrm:N-Rep}]
    Let us fix a density pair $(\sigma,\xi) \in [-1,1] \cross \RR$ and suppose the Gaussian trial state
    \begin{equation} \label{eq:TrialState}
        \psi(\qval) = \begin{pmatrix} c_+ \\ c_- \end{pmatrix}
            \sqrt[4]{\frac{\omega}{\pi}}e^{-\frac{\omega}{2} (\qval - \xi)^2}.
    \end{equation}
    Choosing $c_\pm = \sqrt{(1 \pm \sigma)/2}$ then satisfies all the constraints.
\end{proof}

The state of \cref{eq:TrialState} is not only the prototype for a state with $(\sigma_\psi,\xi_\psi)=(\sigma,\xi)$ but will also reappear as the optimizer for critical $\sigma$ and at zero coupling.

Motivated by \cref{thrm:N-Rep}, let us introduce the \textit{constraint manifold} $\mathcal{M}_{\sigma,\xi}$ collecting all normalized and admissible states that map to a given density pair $(\sigma, \xi) \in [-1,1]\cross \RR$,
\begin{equation*}
    \mathcal{M}_{\sigma,\xi} := \qty{\psi \in Q_0 : \norm{\psi} = 1,\, \sigma_\psi = \sigma,\,   \xi_\psi = \xi }.
\end{equation*}
Then using \cref{thrm:N-Rep,eq:FullHamiltonian} we may for any $(v, j) \in \RR^2$ perform the following partitioning,
\begin{align*}
    E(v,j)
        &= \inf_{\psi \in Q_0} \ev{\hat{H}(v,j)}{\psi} \\[-0.4em]
        &= \inf_{(\sigma,\xi)\in[-1,1] \cross \RR} \qty[  \inf_{\psi \in \mathcal{M}_{\sigma,\xi}} \ev{\hat{H}_0}{\psi} + v\sigma_\psi + j\xi_\psi ] \\
        &= \inf_{(\sigma,\xi)\in[-1,1] \cross \RR} \qty[  F_\mathrm{LL}(\sigma,\xi)  + v\sigma + j\xi ]. \numberthis{\label{eq:LLEnergy}}
\end{align*}
This expression defines the \textit{Levy--Lieb functional} $F_\mathrm{LL}: [-1,1] \cross \RR \to \RR$ by
\begin{equation}\label{eq:DefFLL}
    F_\mathrm{LL}(\sigma,\xi)
        := \inf_{\psi \in \mathcal{M}_{\sigma,\xi}} \ev{\hat{H}_0}{\psi}.
\end{equation}
It then immediately follows that $\abs{F_\mathrm{LL}(\sigma,\xi) } < \infty$. The boundedness below was shown in \cref{eq:lower-bound-ev-of-H}, while the boundedness above is clear from the existence of a trial state. Moreover, for the quantum Rabi model we can prove a wide range of additional properties of $F_\mathrm{LL}$ and its optimizers.

\subsection{Properties of the Levy--Lieb Functional}
We summarize our results in the following theorem.
\begin{theorem}\label{thrm:fllproperties}
    For every density pair $(\sigma,\xi) \in [-1,1]\times \RR $ and any optimizer $\psi$ of $F_\mathrm{LL}(\sigma,\xi)$, the following properties are satisfied.
    \begin{enumerate}
        \item \label{item:fllsymetric}
            $F_\mathrm{LL}(\sigma,\xi) = F_\mathrm{LL}(-\sigma,-\xi)$.
        \item \label{item:flldisplacement}
            For any $\zeta \in \RR$, the displacement rule
            \begin{equation*}
                F_\mathrm{LL}(\sigma,\xi + \zeta)
                    = F_\mathrm{LL}(\sigma,\xi) + \omega^2\zeta \qty( \xi + \frac{\zeta}{2}) + g\sigma\zeta.
            \end{equation*}
            holds, with the special case $\zeta=0$,
            \begin{equation*}
                F_\mathrm{LL}(\sigma,\xi)
                    = F_\mathrm{LL}(\sigma,0)  + g \sigma\xi + \frac{\omega^2}{2}\xi^2.
            \end{equation*}
        \item \label{item:flluniqueoptimiser} An optimizer can always be chosen real and non-negative in both components ($\psi_\pm$).
        \item \label{item:fllvirialrel}
            Any optimizer $\psi$ of $F_\mathrm{LL}(\sigma,0)$ satisfies the virial relation
            \begin{equation*}
                \int  \qty(\frac{1}{2}\abs{\psi'}^2 - \frac{\omega^2}{2} \qval^2\abs{\psi}^2)\dd{\qval}
                  = g\int  \qval\abs{\psi_+}^2 \dd{\qval} .
            \end{equation*}
        \item \label{item:fllkinhopeq}
            For any optimizer
            \begin{align*}
                &\int  \qval \abs{\psi_+}^2 \dd{\qval}
                    = \xi -  \int  \qval \abs{\psi_-}^2 \dd{\qval} \\
                    =& - \frac{2 t}{\omega^2} \int  \psi_+'\psi_- \dd{\qval}
                        - \frac{g(1-\sigma^2)}{2\omega^2} + \frac{\xi (1 + \sigma)}{2} .
            \end{align*}
        \item \label{item:fllkinhopineq}
            Any optimizer satisfies the bound
            \begin{equation*}
                 \int \psi_+''\psi_- \leq \frac{\omega^2}{8t}\qty(1-\sigma^2).
            \end{equation*}
    \end{enumerate}
\end{theorem}

Let us now prove these properties and make some remarks on their relevance. For the sake of readability, the proof of \cref{thrm:fllproperties} has been divided up into five smaller sections including the respective discussions.

\subsubsection*[Proof of Theorem~\ref{thrm:fllproperties}.\ref{item:fllsymetric}]{Proof of Theorem~\ref{thrm:fllproperties}.\ref{item:fllsymetric}: \\ \indent Symmetry of the Levy--Lieb Functional}

Take the joint transformation of parity and time conjugation\footnote{In this case we identify the parity transformation by its usual action of reflection through the origin, i.e., $\mathcal{P}:\qval\mapsto -\qval $. The time conjugation, is identified by its usual action of interchanging polarization/spin, i.e., $\mathcal{T}:\sigma\mapsto -\sigma$, even though the theory presented here is time-independent.} $\mathcal{P}\mathcal{T}$, i.e., $\psi_\pm(\qval) \mapsto \tilde{\psi}_\pm(\qval)= \psi_\mp(-\qval) $. The $ \mathcal{P}\mathcal{T}$ transformation is clearly a unitary transformation that yields
\begin{equation*}
    \norm{\tilde{\psi}_\pm}^2 = \frac{1\mp\sigma}{2} \quad \text{and} \quad \xi_{\tilde{\psi}} = - \xi,
\end{equation*}
thus implying that $(\sigma,\xi) \overset{\mathcal{P}\mathcal{T}}{\longmapsto} (-\sigma,-\xi)$. Furthermore, the transformation leaves the quadratic form of the internal Hamiltonian unchanged, i.e., $\ev*{\hat{H}_0}{\tilde{\psi}} = \ev{\hat{H}_0}{\psi}$. It then follows that the Levy--Lieb functional is \textit{symmetric} in the density pair $(\sigma,\xi)$ for all possible pairs, which concludes the proof.

This symmetry simplifies further studies of the the functional, in particular since \crefpart{thrm:fllproperties}{item:flldisplacement} explicitly determines all dependence in the displacement, $\xi$. The symmetry thus implies that it will always be sufficient to investigate the functional for  $\sigma \in [0,1]$ at one fixed value of $\xi$, typically chosen to be zero. This is especially useful in the numerical investigations, as it effectively halves the search space in $\sigma$.

\subsubsection*[Proof of Theorem~\ref{thrm:fllproperties}.\ref{item:flldisplacement}]{Proof of Theorem~\ref{thrm:fllproperties}.\ref{item:flldisplacement}: \\ \indent Displacement of the Levy--Lieb Functional}
\label{sec:Disp-rule}

Consider the shift operator \(\hat{\mathcal{D}}_\zeta\) which displaces the harmonic-oscillator coordinate by \(\zeta \in \RR\). In particular, $\hat{\mathcal{D}}_\zeta$ maps $\psi_\pm(\qval)$ to $\psi_\pm(\qval - \zeta)$, whilst leaving the two-level decomposition unchanged. It then follows that $\|\hat{\mathcal{D}}_\zeta\psi\|^2 = \|\psi\|^2$, $\sigma_{\hat{\mathcal{D}}_\zeta \psi} = \sigma_{\psi}$, and $\xi_{\hat{\mathcal{D}}_\zeta \psi} = \xi_{\psi} + \zeta$. From a direct computation we have
\begin{equation*}
    \ev*{\hat H_0}_{\hat{\mathcal{D}}_\zeta \psi}
        = \ev*{\hat H_0}_\psi + \omega^2 \zeta \qty( \xi_\psi + \frac{\zeta}{2} ) + g \sigma_\psi \zeta.
\end{equation*}
Then, from the definition of the Levy--Lieb functional we obtain
\begin{align*}
    F_\mathrm{LL}(\sigma, & \xi + \zeta)
        = \inf_{\hat{\mathcal{D}}_\zeta \psi \mapsto (\sigma, \xi + \zeta)} \ev*{\hat H_0}_{\hat{\mathcal{D}}_\zeta \psi} \\
        &= \inf_{\psi \mapsto (\sigma, \xi)} \qty{ \ev*{\hat H_0}_\psi + \omega^2 \zeta \qty( \xi_\psi + \frac{\zeta}{2} ) + g \sigma_\psi \zeta } \\
        &= \inf_{\psi \mapsto (\sigma, \xi)} \ev*{\hat H_0}_\psi + \omega^2 \zeta \qty( \xi + \frac{\zeta}{2} ) + g \sigma \zeta \\
        &= F_\mathrm{LL}(\sigma, \xi) + \omega^2 \zeta \qty( \xi + \frac{\zeta}{2} ) + g \sigma \zeta,
\end{align*}
which is the general displacement relation. Note the second to last equality, in which we restrict the minimization to apply to the first term only. The search for minimizers is performed over the space of wavefunctions which map to a specific set of internal variables \((\sigma, \xi)\), which means that all terms but the first are constant during the minimization process. Further note that the above also shows that if $\psi$ is the optimizer of $ F_\mathrm{LL}(\sigma,\xi)$ then $\hat{\mathcal{D}}_\zeta \psi$ is the optimizer of $F_\mathrm{LL}(\sigma, \xi+\zeta)$. Thus, if the optimizer is known at one $\xi$, all other optimizers in $\xi$ can be obtained by displacement. Unfortunately, the dependence in the polarization is not as simple.

The result in \crefpart{thrm:fllproperties}{item:flldisplacement}, which also holds in the generalization to the Dicke model, is important as it always allows us to explicitly extract all dependency on the displacement \(\xi\).

\subsubsection*[Proof of Theorem~\ref{thrm:fllproperties}.\ref{item:flluniqueoptimiser}]{Proof of Theorem~\ref{thrm:fllproperties}.\ref{item:flluniqueoptimiser}: \\ \indent Real \& Positive Optimizers}
The statement follows from the same argument as the real-valuedness and non-negativity of the components of the ground state given in \cref{subsec:GSProperies}. In particular, the argument also holds for optimizers corresponding to the critical points $\sigma = \pm 1$.

While the real-valuedness also holds for the Dicke model~\cite{Bakkestuen2024}, the non-negativity does not straightforwardly follow. This is important, since the result that optimizers of the Levy--Lieb functional are always ground states (\cref{thrm:OptimizerGS}) relies on this feature.

\subsubsection*[Proof of Theorem~\ref{thrm:fllproperties}.\ref{item:fllvirialrel}]{Proof of Theorem~\ref{thrm:fllproperties}.\ref{item:fllvirialrel}: \\ \indent Virial Relation}

Let us consider the usual an-isotropic coordinate scaling of \citet{LevyPerdewIJQC1994}, that is scaling the harmonic-oscillator coordinate only. For $\mu > 0$ define the transformation $\psi \mapsto \psi^\mu$ of an optimizer by $\psi_\pm^\mu(\qval) = \sqrt{\mu} \psi_{\pm}(\mu \qval)$. By direct calculation we find that 
\begin{equation*}
    \norm{\psi_\pm^\mu}^2 = \frac{1\pm\sigma}{2} \quad \text{and} \quad \xi_{\psi^\mu} = \frac{1}{\mu} \xi_\psi,
\end{equation*} 
and also that 
\begin{align*}
    \expval*{\hat{H}_0}_{\psi^\mu} =&\frac{\mu^2}{2} \expval{\hat{p}^2}_{\psi} + \frac{\omega^2}{2\mu^2} \expval{\qop^2}_{\psi} + \frac{g}{\mu} \expval{\hat{\sigma}_z \qop}_{\psi} -t \expval{\hat{\sigma}_x}_{\psi} .
\end{align*}
Let $\psi$ be an optimizer of $F_\mathrm{LL}(\sigma,0)$, then $\sigma=\sigma_\psi=\sigma_{\psi^\mu}$ and $\xi=\xi_\psi=\xi_{\psi^\mu}=0$ are left unchanged by the scaling. The stationarity condition on the internal energy, \( \dv{\mu} \ev{\hat{H}_0}{\psi^\mu} |_{\mu=1} = 0,\) implies that 
\begin{align*}
    \int \abs{\psi'}^2 - \omega^2 \int \qval^2\abs{\psi}^2 \dd{\qval} = g \int \qval\qty(\abs{\psi_+}^2 - \abs{\psi_-}^2)\dd{\qval}.
\end{align*}
Then, if we recall that $\xi_\psi = \int \qval (\abs{\psi_+}^2 + \abs{\psi_-}^2)\dd{\qval}  =0$, we obtain that 
\begin{align*}
    \int \abs{\psi'}^2 - \omega^2 \int \qval^2\abs{\psi}^2 \dd{\qval} = 2 g \int \qval \abs{\psi_+}^2 ,
\end{align*}
from which \crefpart{thrm:fllproperties}{item:fllvirialrel} follows. The same result was already achieved in \cref{eq:classical-VR} in the context of the hypervirial theorem. This virial relation will be of importance when analyzing the adiabatic connection for the Levy--Lieb functional in \cref{sec:AC}.

\subsubsection*[Proofs of Theorem~\ref{thrm:fllproperties}.(\ref{item:fllkinhopeq}-\ref{item:fllkinhopineq})]{Proofs of Theorem~\ref{thrm:fllproperties}.(\ref{item:fllkinhopeq}-\ref{item:fllkinhopineq}): \\ \indent Relations for the Kinetic Hopping}
Given an optimizer $\psi$, consider the transformation $\psi \mapsto \psi^s$ for some $s \in \RR$ and fixed $\sigma \in [-1,1]$ given by
\[
    \psi_\pm(\qval) \longmapsto \psi^s_\pm (\qval) = \psi_\pm (\qval + (\sigma \mp 1) s).
\]
It then follows from direct calculation that the transformation leaves the constraints unchanged, in particular $\sigma_{\psi^s} = \sigma_\psi = \sigma$ and $\xi_{\psi^s} = \xi_{\psi} = \xi$. Let us then consider $\ev*{\hat{H}_0}_{\psi^s}$, the expectation value of the internal Hamiltonian with respect to the state \(\psi^s\), which contains the following terms
\begin{align*}
    \ev{\hat{p}^2}_{\psi^s} &= \ev{\hat{p}^2}_{\psi},
    \\
    \ev{\qop^2}_{\psi^s}
        &= \ev{\qop^2}_{\psi} + 2 s \qty(\ev{\hat{\sigma}_z \qop}_{\psi} - \sigma \xi)
            + s^2 (1-\sigma^2), 
    \\
    \ev{\hat{\sigma}_x}_{\psi^s}
        &= 2\int  \psi_+(\qval + (\sigma - 1)s) \psi_-(\qval + (\sigma + 1)s) \dd{\qval},
    \\
    \ev{\hat{\sigma}_z \qop}_{\psi^s}
        &= \ev{\hat{\sigma}_z\qop}_{\psi} + s(1-\sigma^2),
\end{align*}
where we already used the result from \crefpart{thrm:fllproperties}{item:flluniqueoptimiser} that the optimizer can be chosen real. The stationarity condition \(  \dv{s} \ev*{\hat{H}_0}_{\psi^s} |_{s=0} = 0\) then implies that
\begin{equation*}
    4 t \int  \psi_+'\psi_-
        = \omega^2 \sigma \xi
            -  \omega^2 \ev*{\hat{\sigma}_z \qop}_{\psi} 
            - g (1-\sigma^2),
\end{equation*}
using integration by parts. Then, using again that $\int  \qval (\abs{\psi_+}^2 + \abs{\psi_-}^2)\dd{\qval} = \xi$ we have that $\ev{\hat{\sigma}_z \qop}_{\psi}  = 2 \int  \qval \abs{\psi_+}^2 \dd{\qval} - \xi$ such that
\begin{equation*}
    4 t \int  \psi_+'\psi_-
        = \omega^2 \xi (1 + \sigma) - g (1-\sigma^2)
            - 2 \omega^2 \int  \qval \abs{\psi_+}^2 \dd{\qval}.
\end{equation*}
\crefpart{thrm:fllproperties}{item:fllkinhopeq} then follows from a slight rearrangement. Note that the same result also follows from using \cref{eq:virial-sigmaz-p,eq:sigma-xi-j-relation} together.

Finally, in order to show \crefpart{thrm:fllproperties}{item:fllkinhopineq}, we consider the second-order condition
\[ \left. \dv[2]{s} \ev*{\hat{H}_0}_{\psi^s} \right|_{s=0} \geq 0. \]
Then only two terms will contribute, and it follows from integration by parts that
\begin{equation*}
    -4t\int  \psi_+''\psi_-  + \frac{\omega^2}{2} (1 - \sigma^2) \geq 0.
\end{equation*}

\subsection{Optimizers are Ground States and \texorpdfstring{$v$}{v}-Representability}\label{sec:optimizers-gs}
In this section we discuss the relation between optimizers of the Levy--Lieb functional and eigenstates for the Schrödinger equation of the quantum Rabi model. Importantly, we will show that optimizers are even ground states and show full $v$-representability if $\sigma\neq\pm 1$. But first, the question arises whether the infimum can always be attained, i.e., if an optimizer even exists. Fortunately, this question can be answered positively.

\begin{theorem}\label{thrm:ExistenceOptimizer}
    For every density pair $(\sigma,\xi) \in [-1,1] \cross \RR$, there exists an optimizer $\psi \in \mathcal{M}_{\sigma,\xi}$ such that
    \begin{equation*}
        F_\mathrm{LL}(\sigma,\xi) = \ev{\hat{H}_0}{\psi},
    \end{equation*}
    i.e., the infimum in \cref{eq:DefFLL} is a minimum.
\end{theorem}

The proof, also for the more general case of the Dicke model, can be found in \citet[Theorem 3.4]{Bakkestuen2024}. Note that within the standard formulation of DFT, the proof of the analogous statement, e.g.\ \citet[Theorem~3.3]{Lieb1983}, relies on the density constraint on the wavefunction to obtain the necessary convergence of the optimizing sequence\footnote{The argument of Lieb in standard DFT has also analogous counterparts in paramagnetic CDFT~\cite{Laestadius2014,Kvaal2021}.}. In the setting of the quantum Rabi model and the Dicke model, this feature can be obtained by the trapping nature of the harmonic-oscillator potential.

For regular values of $\sigma$, the following optimality result is obtained in the language of Lagrange multipliers by use of the first and second order criticality conditions on the quadratic form $\ev{\hat{H}_0}{\psi}$. The use of Lagrange multipliers comes naturally, since we have three constraints to fulfill within the variational problem of \cref{eq:DefFLL}: $\|\psi\|=1$, $\sigma_\psi=\sigma$, and $\xi_\psi=\xi$. The corresponding Lagrange multipliers are then $E$, $v$, and $j$. For the full argument, see the proof of Theorem~3.7 in \citet{Bakkestuen2024}, where the result is shown to hold for all regular $\sigma$ in the generalization to the Dicke model, or Theorem~3.18 in the same reference for a specialization to the quantum Rabi model.

\begin{theorem} \label{thrm:OptimalityRegular}
    Let $(\sigma,\xi)\in (-1,1)\times \RR $ and suppose that $\psi \in \mathcal{M}_{\sigma,\xi}$ is an optimizer of $ F_\mathrm{LL}(\sigma,\xi)$. Then there exist unique Lagrange multipliers $E,v,j\in \RR $ such that $\psi$ satisfies the Schrödinger equation (in the strong sense)
    \begin{equation}\label{eq:SE}
        \hat{H}(v,j) \psi = E \psi, 
    \end{equation}
    as well as the second-order condition
    \begin{equation*}
        \ev{\hat{H}(v,j)}{\chi} \geq E \norm{\chi}^2
    \end{equation*}
    for all $\chi$ in the the tangent space of $\mathcal{M}_{\sigma,\xi}$ at $\psi$, characterized by $\chi\in Q_0$, $\braket{\psi}{\chi} = 0$, $\mel{\psi}{\hat{\sigma}_z}{\chi} = 0$, and $\mel{\psi}{\qop}{\chi} = 0$.
    Furthermore, $\psi$ has the internal energy
    \begin{equation*}
        F_\mathrm{LL}(\sigma,\xi) = \ev{\hat{H}_0}{\psi} = E - v\sigma - j\xi.
    \end{equation*}
\end{theorem}

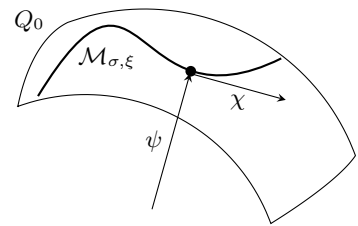
\begin{figure}
    \centering
\usetikzlibrary{shapes.geometric, arrows}
\usetikzlibrary{3d}
\begin{tikzpicture}[sphere segment/.style args={%
phi from #1 to #2 and theta from #3 to #4 and radius #5}{insert path={%
 plot[variable=\x,smooth,domain=#2:#1] 
 (xyz spherical cs:radius=#5,longitude=\x,latitude=#3)
 -- plot[variable=\x,smooth,domain=#3:#4] 
 (xyz spherical cs:radius=#5,longitude=#1,latitude=\x)
 --plot[variable=\x,smooth,domain=#1:#2] 
 (xyz spherical cs:radius=#5,longitude=\x,latitude=#4)
 -- plot[variable=\x,smooth,domain=#4:#3] 
 (xyz spherical cs:radius=#5,longitude=#2,latitude=\x)}}]

   \draw[thin,sphere segment={phi from -20 to 70 and theta from -10 to 30 and radius 3}] ;

  \def\dx{-2};
  \def\dy{3.5};
  
  \draw[thick] ({0.8+\dx},{-1.5+\dy}) .. controls ({2.2+\dx},{0.8+\dy}) and ({2+\dx},{-2+\dy}) .. ({4+\dx},{-1+\dy});
  \draw[->,>=stealth] ({2.3+\dx},{-3+\dy}) -- node[left]{$\psi$} ({2.8+\dx},{-1.2+\dy});
  \draw[->,>=stealth] ({2.8+\dx},{-1.2+\dy}) -- node[below]{$\chi$} ({2.8+1.8*0.7+\dx},{-1.2-0.5*0.7+\dy});
  \fill ({2.81+\dx},{-1.17+\dy}) circle[radius=2pt];
  \node at ({2.2+\dx},{-1+\dy}) [left] {$\mathcal{M}_{\sigma,\xi}$};
  \node at ({1+\dx},{-.5+\dy}) [left] {$Q_0$};
  
\end{tikzpicture}
    \caption{The constraint manifold $\mathcal{M}_{\sigma,\xi}$ inside the set of admissible wavefunctions $Q_0$ with an optimizer $\psi$ and a vector $\chi$ from the tangent space.}
    \label{fig:manifold}
\end{figure}

The situation is sketched in Fig.~\ref{fig:manifold}. Note that while a constraint manifold can be analogously formulated within the standard formulation of DFT, the full result from above cannot be obtained since the density constraint does not give rise to a well-defined tangent space due to problems with differentiability~\cite{Lammert2007}.

Additionally, a similar result can also be formulated for the critical values of $\sigma$, which in the case of the quantum Rabi model are simply the endpoints $\sigma = \pm 1$.

\begin{theorem}\label{thrm:OptimalityIrregular}
    Let $\sigma = \pm 1 $ and $\xi \in \RR $. Suppose that $\psi \in \mathcal{M}_{\sigma,\xi}$ is an optimizer of $ F_\mathrm{LL}(\sigma,\xi)$. Then $\psi_\mp \equiv 0$ and there exists an $n \in \mathbb{N}_0$ such that $\tilde{\psi}_\pm (\qval) = \psi_\pm(\qval + \xi)$ satisfies the harmonic-oscillator equation 
    \begin{equation}\label{eq:OptimalityIrregularSchrodinger}
        -\frac{1}{2}\tilde{\psi}_\pm '' + \frac{\omega^2}{2} \qval^2 \tilde{\psi}_\pm = \omega\qty(n + \frac{1}{2}) \tilde{\psi}_\pm
    \end{equation}
    in the strong sense. The $\psi$ has the internal energy
    \begin{equation}\label{eq:OptimalityIrregularEnergy}
        \ev{\hat{H}_0}{\psi} = \omega\qty(n + \frac{1}{2}) \pm g \xi + \frac{\omega^2}{2}\xi^2.
    \end{equation}
\end{theorem}
The proof of this results relies, similarly to the proof of \cref{thrm:OptimalityRegular}, on the criticality condition of the quadratic form $\ev{\hat{H}_0}{\psi}$. The full derivation can be found in \citet[Theorem~3.18]{Bakkestuen2024}, but note here that this result does not immediately generalize to the Dicke model. In fact, the optimality conditions for the critical $\sigma$'s can also be written down for the general Dicke model, but they are significantly more complicated and do not take the simple form of \cref{eq:OptimalityIrregularSchrodinger,eq:OptimalityIrregularEnergy}. Moreover, the following result follows immediately from \cref{thrm:OptimalityIrregular}, in particular from the harmonic-oscillator solution of \cref{eq:OptimalityIrregularSchrodinger} where, quite obviously, we can limit ourselves to $n=0$ for the optimizer that is thus unique. 
\begin{corollary}\label{corollary:OptimizerIrreg}
    For any $\xi \in \RR $ and critical $\sigma$, the Gaussian trial state from \cref{eq:TrialState}, with $c_+=1,c_-=0$ if $\sigma = +1$ and $c_+=0,c_-=1$ if $\sigma = -1$, is the unique optimizer of $F_\mathrm{LL}(\pm 1,\xi)$. The corresponding values of the Levy--Lieb functional are   
    \[F_\mathrm{LL}(\pm 1,\xi) = \frac{\omega}{2} \pm g\xi + \frac{\omega^2}{2} \xi^2 .\]
\end{corollary}

Can we say something similar for the regular values of $\sigma$? Indeed the optimizers of $ F_\mathrm{LL}(\sigma,\xi)$ are always unique and not only solve the Schrödinger equation but constitute its \emph{ground-state solution}. This important result does unfortunately not generalize fully to the Dicke model. There, we may only conclude that the optimizer is a low-lying eigenstate, in particular at most the $(N+1)$-th lowest eigenstate (see \citet[Theorem~3.9]{Bakkestuen2024}).

\begin{theorem}\label{thrm:OptimizerGS}
    For every density pair $(\sigma, \xi) \in (-1,1)\times\RR $ there exists a unique real-valued and strictly positive optimizer $\psi \in \mathcal{M}_{\sigma,\xi}$ of $F_\mathrm{LL}(\sigma, \xi)$ that is the (non-degenerate) ground-state solution of \cref{eq:SE}.
\end{theorem}

\begin{proof}
    Firstly, we invoke \cref{thrm:ExistenceOptimizer} to show that an optimizer $\psi$ exists and \crefpart{thrm:fllproperties}{item:flluniqueoptimiser} to have it real and with non-negative components. Furthermore, the optimizer $\psi$ satisfies the Schrödinger equation in the strong sense by \cref{thrm:OptimalityRegular}. This means it has all the properties that we showed for ground states in \cref{subsec:GSProperies}, in particular that it is unique.
\end{proof}

The preceding theorem has tremendous consequences. Firstly, since every regular density pair is a ground-state solution for an Hamiltonian with some $(v,j)\in\RR^2$, we achieved a full $v$-representability result. Secondly, since this ground state is always non-degenerate, there is no need to ever rely on mixed states. This means that the Lieb functional $F_\mathrm{L}$~\cite{Lieb1983}, that is equal to the constrained search over ensembles, will give the same result. In other words, the Levy--Lieb functional and the Lieb functional are the same. Finally, since the Lieb functional is always convex, it implies that $F_\mathrm{LL}$ is also convex. Then since the external quantities $(v,j)$ from the subdifferential $-\underline\partial F(\sigma,\xi)$ are unique as Lagrange multipliers (\cref{thrm:OptimalityRegular}), the functional is differentiable by a basic result from convex analysis (see, e.g., Theorem 25.1 in \citet{rockafellar-book}). We collect those results in the following corollary.

\begin{corollary}\label{corollary:OptimizerGS}
    Consider a regular density pair $(\sigma, \xi) \in (-1,1)\times\RR$. Then the following holds:
    \begin{enumerate}
        \item\label{item:vRep} ($v$-representability) The $(\sigma, \xi)$ is uniquely pure-state $v$-representable.
        \item\label{item:FLL-FL-equiv} (equivalence of functionals) $F_\mathrm{LL}(\sigma,\xi) = F_\mathrm{L}(\sigma,\xi)$.
        \item\label{item:deriv} (differentiability) The $F_\mathrm{LL}$ is differentiable at $(\sigma,\xi)$ and $(v,j) = -\nabla F_\mathrm{LL}(\sigma,\xi)$ is its representing external potential pair.
    \end{enumerate}
\end{corollary}

Here, \cref{item:vRep} shows that the mapping $ (-1,1)\times\RR  \ni  (\sigma, \xi) \mapsto (v,j) \in \RR ^2$ is a bijection. This already includes the Hohenberg--Kohn result of injectivity from \cref{thrm:HKRabi}. Yet, if $\sigma=\pm 1$ then $(\sigma,\xi)$ is not $v$-representable. This is also a significantly stronger result than for other settings where $v$-representability is not available. \Cref{item:FLL-FL-equiv} also holds for the critical endpoints at $\sigma=\pm 1$, where the proof, alongside a more detailed proof also for the other points, can be found in \citet{Bakkestuen2024}. For $|\sigma|\to 1$ the derivative from \cref{item:deriv} must go to infinity, $|v|\to\infty$, since the derivative of a convex function must always be monotone.

\subsection[Zero Coupling Functional]{The Levy--Lieb Functional at Zero Coupling}\label{sec:zero-coupling}
For the purposes of the later discussion on the adiabatic connection in \cref{sec:AC}, let us consider the special case of $g=0$. In analogy to the usual quantum chemistry language this can also be referred to as the Kohn--Sham model of the system as it describes a non-interacting (or rather uncoupled) variant of the model at hand. In this case the internal Hamiltonian simplifies to 
\begin{equation}\label{eq:HZeroCoupling}
    \hat{H}_0 = \frac{1}{2} \hat{p}^2 + \frac{1}{2}\omega^2 \qop^2
        - t \hat{\sigma}_x,
\end{equation}
which then naturally is diagonal in the eigenbasis of $\hat{\sigma}_x$. 

Recall from \cref{subsec:GSProperies} that we may approximate the expectation value of \cref{eq:HZeroCoupling} from below by the ground state of the harmonic oscillator and estimate the kinetic coupling term by the Cauchy--Schwarz inequality, such that
\[\ev{\hat{H}_0}{\psi} \geq \frac{\omega}{2} - t\sqrt{1- \sigma^2}. \]
This is of course a lower bound on $F_\mathrm{LL}^0 $ as well, in particular on $F^0_\mathrm{LL}(\sigma,0)$. Note that we use a zero superscript to denote the Levy--Lieb functional in the special case of $g=0$. 
We will generalize this notation further in the next section, where we 
study variable interaction strengths. 
Then by the displacement rule, \crefpart{thrm:fllproperties}{item:flldisplacement}, 
\begin{equation} \label{eq:FLL0LowerBound}
    F^0_\mathrm{LL}(\sigma,\xi) \geq \frac{\omega}{2} - t\sqrt{1- \sigma^2} + \frac{\omega^2}{2}\xi^2.
\end{equation}

Suppose again the Gaussian trial state, \cref{eq:TrialState}, with the coefficients $c_\pm = \sqrt{(1\pm \sigma)/2}$, which as previously shown satisfies the constraints. For this particular choice of trial state, $\psi$, we find that 
\begin{align*}
    \ev{\hat{p}^2}_{\psi} = \frac{\omega}{2}, \quad 
    \ev{\qop^2}_{\psi} = \frac{1}{2\omega} + \xi^2, \quad
    \ev{\hat{\sigma}_x}_{\psi} = \sqrt{1-\sigma^2}.
\end{align*}
Combining this with the constraints, we find that the trial state has the internal energy
\begin{equation*}
    \ev*{\hat{H}_0}_{\psi} = \frac{\omega}{2} -t \sqrt{1-\sigma^2}  + \frac{\omega^2}{2} \xi^2
\end{equation*}
which would be an upper bound on the ground state energy if the trial state is not a ground-state. However, it is in fact exactly equal the lower bound \cref{eq:FLL0LowerBound}, thus implying that the trial state, \cref{eq:TrialState}, is indeed the optimizer of the Levy--Lieb functional at zero coupling, and by \cref{thrm:OptimizerGS} we can formulate the following result. 
\begin{theorem}[Optimizer at Zero Coupling]\label{thrm:FLL0}
    For $g=0$ and $(\sigma,\xi) \in [-1,1]\times\RR $, the Gaussian trial state, \cref{eq:TrialState}, with coefficients $c_\pm = \sqrt{(1\pm \sigma)/2}$ is the optimizer of the Levy--Lieb functional. For all pairs  $(\sigma,\xi) \in [-1,1]\times\RR $ it holds
    \begin{equation*}
        F^0_\mathrm{LL}(\sigma,\xi) = \frac{\omega}{2} - t\sqrt{1- \sigma^2} + \frac{\omega^2}{2}\xi^2.
    \end{equation*}
\end{theorem}
This result will be of importance when we turn to the discussion of the adiabatic connection of the Levy--Lieb functional in \cref{sec:AC}.

Furthermore, since the value of the Levy--Lieb functional is explicitly known in the zero coupling regime, it allows us to calculate the ground-state energy directly using  \cref{eq:LLEnergy}. By virtue of the unique $v$-representability, \crefpart{corollary:OptimizerGS}{item:vRep}, the infimum in \cref{eq:LLEnergy} is indeed a minimum for all regular densities such that the energy can be readily calculated as direct minimization over $(\sigma,\xi) \in (-1,1) \times \RR $. However, as noted in \crefpart{corollary:OptimizerGS}{item:deriv}, the functional is differentiable for regular polarizations. Therefore, we may for each $(\sigma,\xi) \in (-1,1)\times \RR $ calculate the corresponding external pair directly as the derivative of $F_\mathrm{LL}^0 (\sigma,\xi)$. In particular, 
\begin{align*}
    -\pdv{\sigma} F_\mathrm{LL}^0 (\sigma,\xi) &= -t\frac{\sigma}{\sqrt{1-\sigma^2}} = v(\sigma), \\
    -\pdv{\xi} F_\mathrm{LL}^0 (\sigma,\xi) &= -\omega^2 \xi = j(\xi) .       
\end{align*}    
These expressions define the explicit one-to-one Hohenberg--Kohn mapping in this simplified setting. Inverting these expressions, the optimizers of \cref{eq:LLEnergy} are 
\begin{equation*}
    \sigma(v) = - \frac{v}{\sqrt{v^2 + t^2}} \quad \text{and} \quad \xi(j) = - \frac{j}{\omega^2}.
\end{equation*}
Here, the expression for the displacement is the virial relation \cref{eq:sigma-xi-j-relation} again, in the special case of $g=0$. By insertion into \cref{eq:LLEnergy} we have
\begin{align*}
    E^0(v,j)  &= F_\mathrm{LL}^0 (\sigma(v),\xi(j)) + v\sigma(v) + j \xi(j) \\
            &= \frac{\omega}{2} - \sqrt{v^2 + t^2} - \frac{j^2}{2\omega^2}.           
\end{align*}
Note that for the zero-coupling conjugate pair $(F_\mathrm{LL}^0,E^0)$, the Fenchel--Young inequality $E^0 - F_\mathrm{LL}^0\leq v \sigma  + j \xi $ is then fully saturated at the optimizers, as required. This is also directly seen to be true for the extended version of $F_\mathrm{LL}^0$ defined on the whole $\mathbb R^2$ by setting the functional value equal to $+\infty$ whenever $|\sigma|>1$ (i.e., when a density pair is non $N$-representable) or using the Lieb recipe $\sup_{(v,j)\in \mathbb R^2} (E^0(v,j) - v \sigma  - j \xi )$ for every $(\sigma,\xi)\in \mathbb R^2$. We can make one more comment within the framework of convex analysis. For a conjugate pair like $(F_\mathrm{LL}^0,E^0)$, we know that the optimality condition can be expressed as\footnote{Note that this relation is true in general for a conjugate pair. However, for simplicity, in this work we only state this for $\lambda=0$ where all expressions are explicitly known.}
\begin{equation*}
    -(v,j) \in \underline{\partial} F_\mathrm{LL}^0(\sigma,\xi) \,\, \iff \,\, 
    (\sigma,\xi) \in \overline{\partial} E^0(v,j).
\end{equation*}
The same information is encoded in $F_\mathrm{LL}^0$ and $E^0$. In this simplified setting we can interpret the calculations above that took us from $F_\mathrm{LL}^0$ to $E^0$: by first computing the differentials $-\underline{\partial} F_\mathrm{LL}^0$ we can invert these expressions, or said differently, reflect the expressions along $\sigma  = v$ and $\xi  = j$. This mirroring operations give the elements of $\overline{\partial} E^0$ geometrically. Integrating these differentials takes us to the energy expression $E^0$ (up to a constant). In the case of the extended universal functional, that assumes the value $+\infty$, the vertical asymptotes of the differentials are reflected to horizontal ones. The process is illustrated in \cref{fig:conjugate-pair-geometry}. In summary, the zero coupling case, that will paradigmatically serve as a Kohn--Sham system, thus allows for an explicit form for the mapping to external pairs, the ground-state energy, and the universal functional.
\begin{figure}
    \centering
    \includegraphics{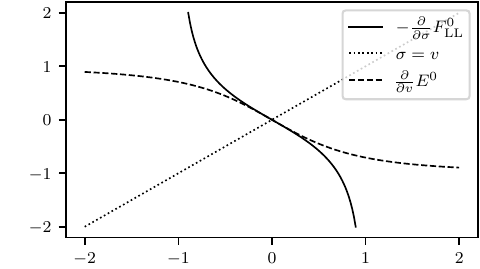}
    \caption{At $\lambda=0$ we can explicitly illustrate how $F_\mathrm{LL}^0$ and $E^0$ encode the same information, in a plot inspired by the work of \citet{Helgaker2022bookchap}. Reflection of the differential $-\pdv{\sigma} F_\mathrm{LL}^0$ along $\sigma =v$ allows us to obtain $\pdv{v} E^0$.}
    \label{fig:conjugate-pair-geometry}
\end{figure}

\section{The Adiabatic Connection}\label{sec:AC}
\subsection{Integral Representation of the Universal Functional}
For the study of the adiabatic In within this DFT formulation of the quantum Rabi model, let us introduce $\lambda\in \RR $ as a scaling for the coupling constant in a similar fashion as in standard DFT, discussed in \cref{sec:DFTBackground}. In this section, let us indicate the dependence on the $\lambda$ by a superscript,  which for the internal Hamiltonian entails, 
\begin{equation*}
    \hat{H}_0^\lambda = \frac{1}{2} \hat{p}^2 + \frac{\omega^2}{2} \qop^2 - t \hat{\sigma}_x + \lambda g \hat{\sigma}_z \qop .
\end{equation*}
The Levy--Lieb functional, denoted $F_\mathrm{LL}^{\lambda}(\sigma,\xi)$, is then given by (for a given $\lambda$)
\begin{equation*}
    F_\mathrm{LL}^\lambda(\sigma,\xi)
        = \inf_{\psi \in \mathcal{M}_{\sigma,\xi}} \ev{\hat{H}_0^\lambda}{\psi} 
\end{equation*}
Thus, $\lambda=0$ gives the uncoupled system while $\lambda=1$ corresponds to the previously considered Hamiltonian with light-matter coupling. Moreover, to show that the function $\RR  \ni \lambda \mapsto F^{\lambda}_\mathrm{LL}(\sigma, \xi)$ is concave for every fixed pair $(\sigma,\xi) \in [-1,1]\times \RR $, choose some $\lambda_1,\lambda_2 \in \RR $ and $s\in[0,1]$. It then follows that 
\begin{align*}
    &F^{s\lambda_1 + (1-s)\lambda_2}_\mathrm{LL}(\sigma, \xi) = \inf_{\psi\mapsto(\sigma,\xi)} \ev{\hat{H}_0^{s\lambda_1 + (1-s)\lambda_2}}{\psi} \\
    &= \inf_{\psi\mapsto(\sigma,\xi)} \ev{s\hat{H}_0^{\lambda_1} + (1-s)\hat{H}_0^{\lambda_2} }{\psi} \\
    &\geq s\inf_{\psi\mapsto(\sigma,\xi)} \ev{\hat{H}_0^{\lambda_1}}{\psi} + (1-s) \inf_{\psi\mapsto(\sigma,\xi)} \ev{\hat{H}_0^{\lambda_2} }{\psi} \\
    &= sF^{\lambda_1}_\mathrm{LL}(\sigma, \xi) + (1-s)F^{\lambda_2}_\mathrm{LL}(\sigma, \xi).
\end{align*}

To obtain the adiabatic connection, we will first calculate the superdifferential of the function $\lambda \mapsto F_\mathrm{LL}^{\lambda}(\sigma,\xi)$ that we defined in \cref{eq:F-superdiff}. Let $\psi^\lambda \in \mathcal{M}_{\sigma,\xi}$ be the optimizer of $F_\mathrm{LL}^\lambda(\sigma,\xi)$ and $\psi^{\lambda'} \in \mathcal{M}_{\sigma,\xi}$ be the optimizer of $F_\mathrm{LL}^{\lambda'}(\sigma,\xi)$. It then follows from the variational principle that 
\begin{align*}
    F_\mathrm{LL}^{\lambda'}(\sigma,\xi) &= \ev*{\hat{H}^{\lambda'}_0}{\psi^{\lambda'}} \leq \ev*{\hat{H}^{\lambda'}_0}{\psi^{\lambda}} \\
    &= \ev*{\hat{H}^{\lambda}_0}{\psi^{\lambda}} + (\lambda'-\lambda) \ev*{g\hat{\sigma}_z\qop}{\psi^{\lambda}} \\
    &= F_\mathrm{LL}^{\lambda}(\sigma,\xi) + (\lambda'-\lambda) g\ev*{\hat{\sigma}_z\qop}{\psi^{\lambda}}.
\end{align*}
We thus have the following result that echoes the standard DFT result from \cref{eq:F-superdiff-trace-W}.
\begin{lemma}\label{lemma:SubdiffElement}
    For every fixed pair $(\sigma,\xi)\in[-1,1]\times \RR $ and $\lambda \in \RR $
    \begin{equation*}
        g \ev{\hat{\sigma}_z \qop}{\psi^\lambda} \in  \overline{\partial}_\lambda F_{\mathrm{LL}}^\lambda(\sigma, \xi),
    \end{equation*}
    where $\psi^\lambda \in \mathcal{M}_{\sigma,\xi}$ is the unique real-valued and strictly positive optimizer of $F_{\mathrm{LL}}^\lambda(\sigma, \xi)$.
\end{lemma}

We have been careful here not to assume that $\lambda\mapsto F_{\mathrm{LL}}^\lambda(\sigma, \xi)$ is differentiable, but actually the non-degeneracy of ground states and the differentiability of $F_{\mathrm{LL}}^\lambda(\sigma, \xi)$ with respect to $\sigma$ and $\xi$ from \crefpart{corollary:OptimizerGS}{item:deriv} is a strong indication that the functional is also differentiable with respect to $\lambda$.

Before applying the Newton--Leibniz formula, \cref{eq:NewtonLeibniz},  to obtain an integral representation of $F_\mathrm{LL}^\lambda$, let us first employ the displacement rule, \crefpart{thrm:fllproperties}{item:flldisplacement}, 
\begin{equation*}
    F_\mathrm{LL}^\lambda(\sigma,\xi) = F_\mathrm{LL}^\lambda(\sigma,0)  + \lambda g \sigma\xi + \frac{\omega^2}{2}\xi^2.
\end{equation*}
The Newton--Leibniz formula applied to $F_\mathrm{LL}^\lambda(\sigma,0)$ using \cref{lemma:SubdiffElement} as the choice for the element of the superdifferential then yields
\begin{equation}\label{eq:AC1}
    \begin{aligned}
        F_\mathrm{LL}^\lambda(\sigma,\xi) =& \, F_\mathrm{LL}^0(\sigma,0) + \lambda g \sigma\xi + \frac{\omega^2}{2}\xi^2  \\
        &+  g\int_0^\lambda \ev{\hat{\sigma}_z \qop}{\varphi^\nu} \dd{\nu},
    \end{aligned}
\end{equation}
where $\varphi^\nu \in \mathcal{M}_{\sigma,0}$ denotes the optimizer of $F_\mathrm{LL}^\nu(\sigma,0)$. Importantly, as the $\varphi^\nu$ do not depend on $\xi$, the above integrand is also independent of $\xi$. Moreover, recall from  \crefpart{thrm:fllproperties}{item:fllkinhopeq} that the integrand in \cref{eq:AC1} can be rewritten as 
\begin{align*}
    \ev{\hat{\sigma}_z\qop}{\varphi^\nu} &= 2 \int \qval\abs{\varphi_+^\nu}^2\dd{\qval} \\
    &= - \frac{4 t}{\omega^2} \int {\varphi^\nu_+}'\varphi_-^\nu \dd{\qval} - \frac{\nu g(1-\sigma^2)}{\omega^2} .
\end{align*}
Using this form, along with the explicit expression $F^0_\mathrm{LL}(\sigma,0)= \frac{\omega}{2} - t\sqrt{1- \sigma^2}$ from \cref{thrm:FLL0}, we obtain the following result.

\begin{theorem}\label{thrm:AC}
    For every $\lambda\in \RR _+$ the Levy--Lieb functional along the adiabatic connection, $F_\mathrm{LL}^\lambda : [-1,1]\times \RR  \to \RR $, is given by 
    \begin{equation*}
        \begin{aligned}
             F_\mathrm{LL}^\lambda(\sigma,\xi) =&\; \frac{\omega}{2} - t\sqrt{1- \sigma^2} 
              + \frac{\omega^2}{2}\xi^2 + \lambda g \sigma\xi \\
              &-\frac{ \lambda^2  g^2 }{2\omega^2} (1-\sigma^2)
              + I^\lambda(\sigma),
        \end{aligned}
    \end{equation*}
    where the only non-explicit contribution is
    \begin{equation*}
        I^\lambda(\sigma) := -\frac{4t g}{\omega^2} \int_0^\lambda \int  {\varphi^\nu_+}'\varphi_-^\nu \dd{\qval}\dd{\nu}.
    \end{equation*}
\end{theorem}

\begin{figure}[ht]
  \centering
  \includegraphics{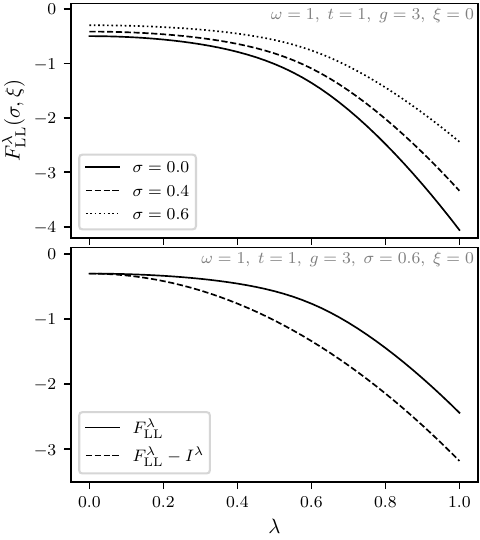}
  \caption{The Levy--Lieb functional as a function of \(\lambda\) for \(\omega = 1\), \(t = 1\), \(g = 3\), and \(\xi = 0\). It is obtained numerically as the convex conjugate of the energy functional. The upper panel shows the value of the functional in $\lambda$ for three values of \(\sigma\) and $\xi=0$. The lower panel shows the value of the functional in $\lambda$ for \(\sigma = 0.6\) and $\xi=0$ (solid line) alongside the explicit terms in the expression for \(F_\mathrm{LL}\) in \cref{thrm:AC} (dashed line), thus depicting the dependence of \(F_\mathrm{LL}\) on the non-explicit correlation energy \(I^\lambda\).}
  \label{fig:the-analytic-components-of-adiabatic-connection}
\end{figure}

In order to visualize the adiabatic connection, we computed the Levy--Lieb functional in $\lambda$ for different values of $\sigma$, see \cref{fig:the-analytic-components-of-adiabatic-connection} (top panel).  Moreover, we compared the full Levy--Lieb functional $F_\mathrm{LL}^\lambda(\sigma,\xi)$ to the explicitly known terms by computing $F_\mathrm{LL}^\lambda(\sigma, \xi) - I^\lambda(\sigma)$, see \cref{fig:the-analytic-components-of-adiabatic-connection} (bottom panel). This shows, as expected, that the Levy--Lieb functional is concave and decreases with $\lambda$. Additionally, as \cref{fig:scaling-of-kinetic-correlation-function-in-lambda} shows more clearly, the non-explicit term $I^\lambda(\sigma)$ is in a positive contribution to the total Levy--Lieb functional that is growing in $\lambda$.

\begin{figure}[ht]
  \centering
  \includegraphics{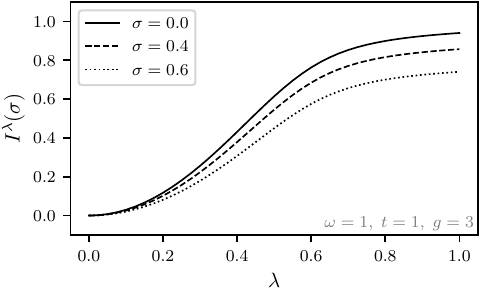}
  \caption{The non-explicit contribution \(I^\lambda(\sigma)\) to the adiabatic connection as a function of \(\lambda\) for three values of \(\sigma\) using \(\omega = 1\), \(t = 1\), and \(g = 3\).}
  \label{fig:scaling-of-kinetic-correlation-function-in-lambda}
\end{figure}

The careful Reader might wonder why we first employed the displacement rule for $F_\mathrm{LL}^\lambda$ and then used the adiabatic connection that connects the zero coupling to any value of $\lambda$ for $F_\mathrm{LL}^\lambda(\sigma,0)$. This is just to obtain the simplest possible expression, since following this route the optimizers are determined at $\xi=0$. Using instead the integral representation of $F_\mathrm{LL}^\lambda(\sigma,\xi)$ (i.e., never invoking the displacement rule), one instead obtains
\begin{equation}\label{eq:ACpsilambda}
    \begin{aligned}
        F_\mathrm{LL}^\lambda(\sigma,\xi) &= F_\mathrm{LL}^0(\sigma,\xi) + \lambda g \sigma\xi \\
        &+ g \int_0^\lambda ( \ev{\hat{\sigma}_z \qop}{\psi^\nu} - \sigma\xi )\dd{\nu} ,
    \end{aligned}
\end{equation}
where $\psi^\nu \in \mathcal{M}_{\sigma,\xi}$ denotes the optimizer of $F_\mathrm{LL}^\nu(\sigma,\xi)$. 
Consequently, we have
\begin{equation}\label{eq:DisplacedCoupling}
    \ev{\hat{\sigma}_z \qop}{\psi^\lambda}=  \sigma\xi + \ev{\hat{\sigma}_z \qop}{\varphi^\lambda}, 
\end{equation}
relating the optimizers at an arbitrary $\xi$ to those at $\xi=0$. This relation is also directly seen from the fact noted in \cref{sec:Disp-rule} that if $\psi$ is the optimizer of $ F_\mathrm{LL}(\sigma,\xi)$ then $\hat{\mathcal{D}}_\zeta \psi$ is the optimizer of $F_\mathrm{LL}(\sigma, \xi+\zeta)$.

Within this basic QEDFT formulation, we thus are able to formulate an almost explicit form (i.e., in terms of the ``density'' variables only) of the adiabatic connection, where the only non-explicit term is  $I^\lambda(\sigma)$, which further depends on the optimizers $\varphi^\nu$ for $\nu \in [0,\lambda]$. This is in contrast to the standard DFT setting in which the adiabatic connection remains entirely non-explicit, see~\cref{eq:ACDFT}.

\subsection{Correlation Contributions}\label{subsec:ExchangeCorrelation}
To continue the study of $F_\mathrm{LL}^\lambda$, we will divide it into different contributions. We begin by noting that the expression given in \cref{eq:ACpsilambda} is in full analogy with the adiabatic connection in standard DFT (see \cref{eq:ACDFT}). Motivated by this structure, let us identify the direct-coupling term as 
\begin{equation}\label{eq:DirectTerm}
       D(\sigma,\xi) := g\sigma\xi.
\end{equation} 
In analogy to standard DFT we identify the exchange-correlation term
\begin{equation}\label{eq:ExchangeCorrelation}
    \lambda G^\lambda(\sigma) := F_\mathrm{LL}^\lambda(\sigma,\xi) - F_\mathrm{LL}^0(\sigma,\xi) - \lambda D(\sigma, \xi).
\end{equation}
Naïvely, this term should also depend on $\xi$, however,  using \cref{eq:ACpsilambda}, we find that for a fixed  $g\in\RR$ and $\lambda >0$, 
\begin{align*}
    \lambda G^\lambda(\sigma) &=  \int_0^\lambda ( g\ev{ \hat{\sigma}_z \qop}{\psi^\nu} - D(\sigma,\xi ))\dd{\nu} \\
    &= g \int_0^\lambda \ev{\hat{\sigma}_z \qop}{\varphi^\nu} \dd{\nu},\numberthis{\label{eq:CorrelationFunctional}}
\end{align*}
where the last equality follows from \cref{eq:DisplacedCoupling}. It is thus clear that the term $G^\lambda$ is independent of $\xi$.
Since this expectation value of the coupling is also the integrand in the adiabatic connection, \cref{eq:AC1}, it is interesting to compare this plot to other DFT settings, where an unproven conjecture says that such adiabatic-connection curves must always be convex, see e.g.\ \citet{crisostomo2023seven} (Section~3). While this conjecture was formulated for usual particle interactions, it clearly does not hold in case of the quantum Rabi model as shown in \cref{fig:expectation-value-of-the-coupling} (top panel).

\begin{figure}[ht]
  \centering
  \includegraphics{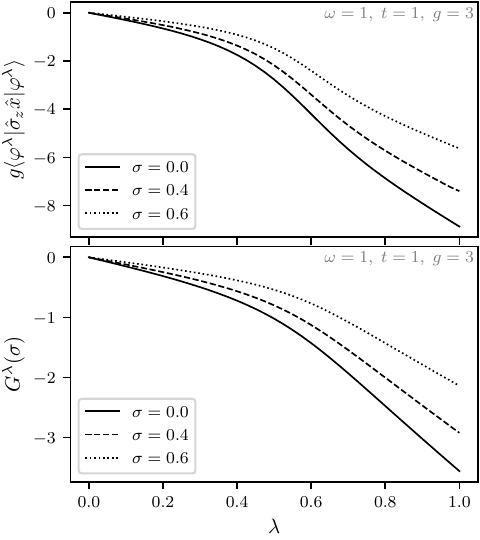}
  \caption{The upper panel shows the expectation value of the coupling, \(\ev*{g \hat \sigma_z \qop}\), for the optimizer \(\varphi^\lambda\) at different coupling strengths \(\lambda\) for three choices of \(\sigma\). The lower panel shows the correlation functional \(G^\lambda(\sigma)\) as a function of \(\lambda\), also for three different values of \(\sigma\). The other parameters are set to \(\omega = 1\), \(t = 1\), and \(g = 3\) for both plots.}
  \label{fig:expectation-value-of-the-coupling}
\end{figure}

However, if we instead write out \cref{eq:ExchangeCorrelation} in terms of the expectation values with respect to $\varphi^\lambda$ and $\varphi^0$ (using first the displacement rule), we obtain an alternative characterization,
\begin{equation*}
    \lambda G^\lambda(\sigma) = P_\mathrm{c}^\lambda(\sigma) + T_\mathrm{c}^\lambda(\sigma) + \lambda W_\mathrm{c}^\lambda(\sigma).
\end{equation*}
Here, we defined the photonic correlation term,
\begin{equation*}
    P_\mathrm{c}^\lambda(\sigma) := \frac{1}{2} \ev*{\hat{p}^2}{\varphi^\lambda} + \frac{\omega^2}{2} \ev*{\qop^2}{\varphi^\lambda} - \frac{\omega}{2},
\end{equation*}
the kinetic correlation term, 
\begin{equation*}
   T_\mathrm{c}^\lambda(\sigma) := (-t)\ev*{\hat{\sigma}_x}{\varphi^\lambda} - (-t)\ev*{\hat{\sigma}_x}{\varphi^0},
\end{equation*}
and the correlation from the coupling,
\begin{equation*}
  \lambda W_\mathrm{c}^\lambda(\sigma) :=  \lambda g \ev*{\hat{\sigma}_z\qop}{\varphi^\lambda}.
\end{equation*}
Then, using \cref{thrm:AC}, or alternatively  \crefpart{thrm:fllproperties}{item:fllkinhopeq} and \cref{thrm:FLL0}, we find that 
\begin{equation*}
    \lambda G^\lambda(\sigma) = I^\lambda(\sigma) - \frac{\lambda^2 g^2}{2\omega^2} (1-\sigma^2) = \int_0^\lambda W_\mathrm{c}^\nu (\sigma) \dd{\nu}.
\end{equation*}
This result corresponds to the perhaps surprising representation of the total exchange-correlation energy in terms of an integral using only the interaction operator in the integrand in standard DFT (while the kinetic correlation energy is still accounted for). 
We can then express the Levy--Lieb functional at $\lambda$ as  
\begin{align*}
    & F_\mathrm{LL}^\lambda(\sigma,\xi) \\
    &= \frac{\omega}{2} - t\sqrt{1-\sigma^2} + \frac{\omega^2 \xi^2}{2} + \lambda g \sigma \xi + \lambda G^\lambda(\sigma) \\
    &= P^0(\xi) + T^0(\sigma)  + \lambda D(\sigma,\xi) + 
        P_\mathrm{c}^\lambda(\sigma) + T_\mathrm{c}^\lambda(\sigma) + \lambda W^\lambda_\mathrm{c}(\sigma).
\end{align*}
Here, $P^0(\xi) = (\omega + \omega^2 \xi^2)/2$ is the zero-coupling photon energy and $T^0(\sigma) = -t\sqrt{1-\sigma^2}$ is the zero-coupling kinetic energy. The total kinetic functional then is 
\begin{equation*}
    T^\lambda(\sigma) := T^0(\sigma) + T_\mathrm{c}^\lambda(\sigma) = -t \ev*{\hat{\sigma}_x}{\varphi^\lambda}.
\end{equation*}
Moreover, recall from  \crefpart{thrm:fllproperties}{item:fllvirialrel}, that
\begin{align*}
    \nu g \ev{\hat{\sigma}_z\qop}{\varphi^\nu} &= 2\nu g \int  \qval\abs{\varphi_+^\nu}^2\dd{\qval} \\
    &= \int \qty(\abs{{\varphi^\nu}'}^2 - \omega^2 \qval^2\abs{\varphi^\nu}^2 )\dd{\qval},
\end{align*}
as $\varphi^\nu$ is the optimizer of $F_\mathrm{LL}^\nu(\sigma,0)$.
By insertion into \cref{eq:CorrelationFunctional}, we thus obtain the alternative characterization 
\begin{equation*}
    \lambda G^\lambda(\sigma) = \int_0^\lambda \frac{1}{\nu} \int\qty(\abs{{\varphi^\nu}'}^2 - \omega^2 \qval^2 \abs{\varphi^\nu}^2 )\dd{\qval} \dd{\nu}.
\end{equation*}

In order to continue the discussion of $G^\lambda(\sigma)$, recall Perdew's definition of exchange energy in \cref{sec:DFTBackground}. Since there is no such thing as a high-density limit in the setting of the quantum Rabi model, we rely on \cref{eq:DefExhange}, where the exchange energy is given as the right derivative of the adiabatic functional at zero coupling minus the direct-coupling term, \cref{eq:DirectTerm},
\begin{align*}
    E_\mathrm{x}(\sigma,\xi) := \lim_{\lambda \to 0^+} \frac{F_\mathrm{LL}^{\lambda}(\sigma,\xi) - F_\mathrm{LL}^{0}(\sigma,\xi)}{\lambda} - D(\sigma,\xi).
\end{align*}
Let the right derivative with respect to $\lambda$ be denoted by $\partial_\lambda^+$. Using \cref{thrm:AC}, we have that
\begin{align*}
    \partial_\lambda^+  F_\mathrm{LL}^{\lambda}(\sigma,\xi) = g \sigma\xi - \frac{\lambda g^2}{\omega^2} (1-\sigma^2) +  \partial_\lambda^+ I^\lambda(\sigma). 
\end{align*}
However, by the definition of $I^\lambda(\sigma)$,
\begin{align*}
    \partial_\lambda^+ I^\lambda(\sigma) =& - \partial_\lambda^+ \qty(\frac{4t g}{\omega^2} \int_0^\lambda \int {\varphi^\nu_+}'\varphi_-^\nu \dd{\qval}\dd{\nu})\\
    =&-\frac{4t g}{\omega^2} \int  {\varphi^\lambda_+}'\varphi_-^\lambda \dd{\qval},
\end{align*}
where we recall that $\varphi^\lambda$ is the optimizer of $F_\mathrm{LL}^\lambda(\sigma,0)$. Then, by setting $\lambda=0$, as required for the definition of the exchange energy, we can readily use the optimizer at zero coupling, \cref{thrm:FLL0}, for $\xi=0$. We thus obtain that 
\begin{align*}
    \left.\partial_\lambda^+ I^\lambda(\sigma) \right|_{\lambda=0} =& -\frac{4t g}{\omega^{3/2}} \frac{1-\sigma^2}{2\sqrt{\pi}} \int u e^{-u^2} \dd{u} =0,
\end{align*}
since the integrand is odd. Consequently, all terms of the functional $G^\lambda(\rho)$ are correlation terms, as summarized in the following theorem.

\begin{theorem}\label{thrm:XC}
    For every density pair $(\sigma, \xi) \in [-1,1]\times \RR $ the \textit{exchange energy} is zero,
    \begin{equation*}
        E_\mathrm{x}(\sigma,\xi) = 0,
    \end{equation*}
    and 
    the \textit{correlation energy} is 
    \begin{align*}
        \lambda G^\lambda(\sigma) 
        =  -\frac{\lambda^2 g^2}{2\omega^2} (1-\sigma^2) +  I^\lambda(\sigma) .
    \end{align*}
\end{theorem}
The absence of exchange energy is not a surprising feature for the quantum Rabi model, and neither would be for the Dicke model, since it is two different components, light and matter, that are coupled here and ``exchange'' only exists between identical, fermionic particles. It is thus clear that the term $G^\lambda(\sigma)$ is purely a correlation term. This also motivates our notation $\lambda G^\lambda = P_\mathrm{c}^\lambda + T_\mathrm{c}^\lambda  + \lambda W_\mathrm{c}^\lambda$.

\subsection{Bounds on Correlation}\label{subsec:SC}
From \cref{thrm:XC} we note that the only non-explicit term in the correlation energy $\lambda G^\lambda (\sigma)$ is $I^\lambda(\sigma)$.  In order to obtain bounds on the correlation energy, a further analysis of this term is warranted.

To establish a lower bound on the Levy--Lieb functional, suppose that $\psi\in Q_0$ is the ground-state solution of $\hat{H}^\lambda(v,j)$ with displacement $\xi\in \mathbb{R}$  and polarization $\sigma \in (-1,1)$. Then by \cref{eq:sigma-xi-j-relation}, $-j = \omega^2\xi + \lambda g \sigma$. Recall from \cref{sec:QRabiModel}, in particular the step just before \cref{eq:lower-bound-ev-of-H}, that 
\begin{align*}
    &\ev*{\hat{H}^\lambda(v,j)}{\psi} \\
        &\quad \geq \frac{\omega}{2} + \sigma \left( v - \frac{j \lambda g}{\omega^2}\right) - t \sqrt{1 - \sigma^2} - \frac{j^2 + \lambda^2 g^2}{2 \omega^2} \\
        &\quad = \frac{\omega}{2} - t \sqrt{1 - \sigma^2} + \frac{\lambda^2 g^2 (\sigma^2 - 1)}{2 \omega^2} - \frac{\omega^2}{2} \xi^2 + v \sigma \\
        &\quad = \frac{\omega}{2} - t \sqrt{1 - \sigma^2} + \frac{\lambda^2 g^2 (\sigma^2 - 1)}{2 \omega^2} + \frac{\omega^2}{2} \xi^2 \\
            &\quad \quad + \lambda g \sigma \xi + v \sigma + j \xi.
\end{align*}
However, since $\psi\mapsto (\sigma,\xi)$ then 
\begin{equation*}
    \ev*{\hat{H}_0^\lambda}{\psi} = \ev*{\hat{H}^\lambda (v,j)}{\psi} - v\sigma - j \xi
\end{equation*}
and we obtain the following lower bound for the Levy--Lieb functional
\begin{equation*}
    F_\mathrm{LL}^\lambda(\sigma,\xi) \geq \frac{\omega}{2} - t\sqrt{1-\sigma^2} + \frac{\omega^2}{2} \xi^2  + \lambda g \sigma \xi - \frac{\lambda^2 g^2}{2\omega^2}(1-\sigma^2). 
\end{equation*}
By taking again \cref{eq:TrialState} as a trial state for $F_\mathrm{LL}(\sigma,\xi)$, we also obtain the upper bound 
\begin{equation*}
  F_\mathrm{LL}^\lambda(\sigma, \xi)
    \leq \frac{\omega}{2}
      + \frac{\omega^2}{2} \xi^2
      - t \sqrt{1 - \sigma^2}
      + \lambda g \sigma \xi.
\end{equation*}
Combining these bounds,  we immediately have 
\begin{equation}\label{eq:I-estimate}
    0 \leq I^\lambda(\sigma) \leq \frac{\lambda^2 g^2}{2\omega^2}(1-\sigma^2).
\end{equation}
For the correlation energy this estimate takes the form of a Lieb--Oxford bound\footnote{We remind the reader that the Lieb--Oxford bound\cite{LiebOxford} (of standard DFT) states that the indirect Coulomb energy of a normalized wavefunction $\psi$ is bounded below by $-C\int_{\mathbb R^3}\rho_\psi^{4/3}$, $C>0$. The bound presented here is similar in spirit to the one-dimensional version~\cite{LaestadiusFaulstich2020}.}.

\begin{proposition}
    The correlation energy satisfies the Lieb--Oxford-type bound 
    \begin{equation*}
       0 \geq  \lambda G^\lambda (\sigma) \geq - \frac{\lambda^2 g^2}{2\omega^2} (1-\sigma^2).
    \end{equation*}
\end{proposition}

Furthermore, in order to get an alternative upper bound for $F_\mathrm{LL}^\lambda(\sigma,\xi)$, suppose another Gaussian trial state
\begin{equation}\label{eq:LambdaTrialState}
    \psi_\pm^\lambda (\qval) = \sqrt{\frac{1\pm \sigma}{2}} \sqrt[4]{\frac{\omega}{\pi}} e^{-\frac{\omega}{2} \qty(\qval - \xi - \frac{\lambda g}{\omega^2}\qty(\sigma \mp 1) )^2},
\end{equation}
which satisfies the necessary constraints
\begin{align*}
    \norm{\psi_\pm^\lambda }^2 = \frac{1\pm\sigma}{2} \quad \text{and} \quad \ev*{\qop}{\psi^\lambda } = \xi.
\end{align*}
By the displacement rule, \crefpart{thrm:fllproperties}{item:flldisplacement}, we may for the sake of simplicity restrict ourselves to the case of $\xi =0$. Then by a direct calculation, we obtain that
\begin{align*}
    &\ev{\hat{p}^2}_{\psi^\lambda} = \frac{\omega}{2}, \quad
        \ev{\qop^2}_{\psi^\lambda} = \frac{1}{2 \omega} + \frac{\lambda^2 g^2}{\omega^4}(1-\sigma^2), \\
    &\ev{\hat{\sigma}_z\qop}_{\psi^\lambda} = - \frac{\lambda g}{\omega^2} (1-\sigma^2), \quad \\
    &\ev{\hat{\sigma}_x}_{\psi^\lambda} = \sqrt{1-\sigma^2}e^{-\frac{\lambda^2g^2}{\omega^3}}.
\end{align*}
Consequently, we have
\begin{align*}
    F_\mathrm{LL}^\lambda(\sigma,0) \leq \frac{\omega}{2} - \frac{\lambda^2 g^2}{2\omega^2} (1-\sigma^2) - t\sqrt{1-\sigma^2} e^{-\frac{\lambda^2g^2}{\omega^3}}.
\end{align*}
However, this upper bound is very similar to the lower bound given above. In particular, the bounds differ only by the exponential $\exp(-\lambda^2 g^2 / \omega^3)$, a factor that rapidly goes to zero as $\lambda$ increases. Conversely, when $\lambda \ll \omega^{3/2}/g$ the exponential is almost one, implying that the upper bound is almost equal to the lower bound. Thus, for $\lambda \ll \omega^{3/2}/g$ the trial state, \cref{eq:LambdaTrialState}, is almost the optimizer of $F_\mathrm{LL}^\lambda(\sigma,\xi)$

By use of the new upper bound, an alternative estimate of ``kinetic'' type follows,
\begin{equation}\label{eq:I-estimate-kinetic}
    0 \leq I^\lambda(\sigma) \leq t \sqrt{1-\sigma^2} \qty(1-e^{-\frac{\lambda^2g^2}{\omega^3}}).
\end{equation}
We can then clearly see that the trial state \cref{eq:LambdaTrialState} is the exact optimizer in the cases \(\sigma = \pm 1\), \(\lambda g = 0\), and \(t = 0\) as well as \(\omega \to \infty\). The first two of these observations are in accordance with the results in \cref{sec:LL}, that the trial state \cref{eq:TrialState} is exact in two cases, in the zero coupling case, \cref{thrm:FLL0}, and for critical polarizations, \cref{corollary:OptimizerIrreg}.

\subsection{Approximate Correlation}\label{subsec:ACresults}
Having established some analytical bounds on the non-explicit correlation term, let us further investigate the term numerically. In particular, \cref{fig:scaling-of-kinetic-correlation-function-in-sigma} shows the non-explicit correlation term as a function of the polarization $\sigma$ for selected values of $\lambda$. A numerical investigation for a decreasing sequence of $\lambda$ shows that $I^\lambda(\sigma)$ indeed approaches zero from above as $\lambda \to 0$, as required by the estimates \cref{eq:I-estimate,eq:I-estimate-kinetic} and as also suggested by \cref{fig:scaling-of-kinetic-correlation-function-in-lambda}. Moreover, we remark that the shape of $I^\lambda(\sigma)$ plotted in $\sigma$ closely resembles the line segment of an ellipse, see \cref{fig:scaling-of-kinetic-correlation-function-in-sigma}. 

\begin{figure}[ht]
  \centering
  \includegraphics{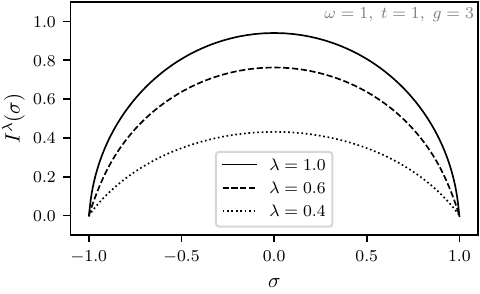}
  \caption{The non-explicit correlation term \(I^\lambda(\sigma)\) as a function of \(\sigma\) for three values of \(\lambda\) using \(\omega = 1\), \(t = 1\), and \(g = 3\).}
  \label{fig:scaling-of-kinetic-correlation-function-in-sigma}
\end{figure}

Motivated by this resemblance we make the following ansatz. Suppose three functions (over the parameters $\lambda$ and $t$) $a$, $b$, and $d$ such that 
\begin{equation*}
    I^\lambda(\sigma) \approx \omega \qty[ b(\lambda,t) \sqrt{1-\frac{\sigma^2}{a(\lambda,t)^2}} - d(\lambda,t)].
\end{equation*}
Then by the constraint $I^\lambda(\pm 1) = 0 $, we have that 
\begin{equation*}
    d(\lambda,t) = b(\lambda, t) \sqrt{1 - \frac{1}{a(\lambda, t)^2}}.   
\end{equation*}
Equipped with this form, let us formulate the following conjecture.

\begin{conjecture}\label{conjecture:Ellipse}
    For $\sigma \in [-1,1]$ the non-explicit correlation functional is of the approximate form
    \begin{equation*}
        I^\lambda(\sigma) \approx \omega \frac{ b(\lambda,t)}{a(\lambda,t)} \qty[ \sqrt{a(\lambda,t)^2 - \sigma^2} -  \sqrt{a(\lambda,t)^2 - 1}].
    \end{equation*}
    Here $a$ and $b$ are functions of the parameters $\lambda$ and $t$.
\end{conjecture}

To further investigate \cref{conjecture:Ellipse}, let us numerically calculate $I^\lambda(\sigma)$ for many combinations of the parameters $\lambda$ and $t$, for $\lambda,t\in[0,3]$ at $\omega=g=1$. Then by performing a parameter fitting, we obtain $a$ and $b$ as functions of parameters $\lambda$ and $t$ as shown in \cref{fig:ellipse-params-in-lambda,fig:ellipse-params-in-t}. Importantly, the parameter fitting shows that \cref{conjecture:Ellipse} fits very well with the numerical simulations. In fact, the largest standard deviation in the parameter fitting is only 0.05 (arising in $b$).

\begin{figure}[htb]
    \centering
    \includegraphics{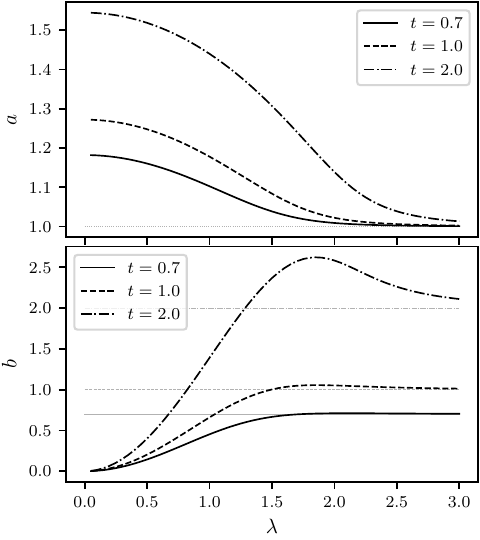}
    \caption{The parameter fitting of the functions $a$ and $b$ in $\lambda$ for three select values of $t$ at  \(\omega = g = 1\).}
    \label{fig:ellipse-params-in-lambda}
\end{figure}

\begin{figure}[htb]
    \centering
    \includegraphics{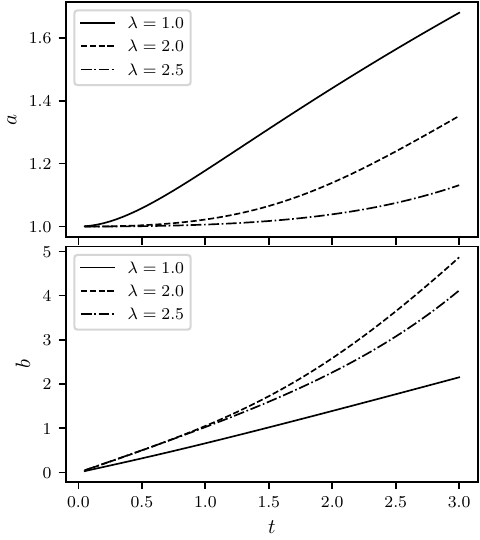}
    \caption{The parameter fitting of the functions $a$ and $b$ in $t$ for three select values of $\lambda$ at  \(\omega = g = 1\).}
    \label{fig:ellipse-params-in-t}
\end{figure}

From \cref{fig:ellipse-params-in-lambda,fig:ellipse-params-in-t}, we learn the following about the functions $a$ and $b$.
\begin{enumerate}
    \item $\lim\limits_{\lambda\to \infty} a(\lambda,t) = 1$.
    \item $\lim\limits_{\lambda\to \infty} b(\lambda,t) = t$
    \item $\lim\limits_{\lambda\to 0} b(\lambda,t) = 0$ (required by \cref{eq:I-estimate,eq:I-estimate-kinetic})
    \item  $\lim\limits_{t\to 0} a(\lambda,t) = 1$ 
    \item  $\lim\limits_{t\to 0} b(\lambda,t) = 0$ 
\end{enumerate}
This shows, in particular, that in the strictly correlated regime ($\lambda=\infty$) the non-explicit part of the correlation functional is
\begin{equation*}
    I^\infty(\sigma) = t\sqrt{1-\sigma^2},
\end{equation*}
i.e., saturates the upper bound of \cref{eq:I-estimate-kinetic}. This implies that near the strictly correlated regime, the Levy--Lieb functional is 
\begin{equation*}
    F_\mathrm{LL}^\mathrm{sc} (\sigma,\xi) \approx \frac{\omega}{2} + \frac{\omega^2}{2}\xi^2 + \lambda g \sigma\xi - \lambda^2 \frac{g^2 }{2\omega^2} (1-\sigma^2).
\end{equation*}
Thus in the strictly correlated regime, there are no kinetic contributions. This is not unexpected, since the coupling term and the displacement operator are unbounded operators whilst the kinetic term is bounded by $|t|$.

\section{Photon-free Approximation}
\label{sec:pf}

\subsection{Effective Potential}
\label{subsec:eff-pot}

In order to derive the \textit{photon-free} approximation, recall the hypervirial relation of \cref{eq:force-balance}. This relation can be established  separately for a non-coupled auxiliary system at $\lambda=0$ ($g=0$) with ground-state wavefunction $\phi$, which was studied in \cref{sec:zero-coupling}, and the fully coupled system at $\lambda=1$ with the ground state $\psi$. In these two cases, \cref{eq:force-balance} gives
\begin{subequations}\label{eq:FB-full-aux}
    \begin{align}
        &t\sigma + g\ev{\hat\sigma_x\qop}{\psi} + v\ev{\hat\sigma_x}{\psi} = 0, \\
        &t\sigma + v_\s\ev{\hat\sigma_x}{\phi} = 0.
    \end{align}
\end{subequations}
As it is usual in DFT, the potential $v_\s$ for the auxiliary system (which, again in analogy to standard DFT would be called Kohn--Sham system) is chosen such that the value of $\sigma$ agrees for both systems. This is surely possible if $\sigma\neq\pm 1$ because of the $v$-representability result from \crefpart{corollary:OptimizerGS}{item:vRep}. In a similar fashion, the values for $\xi$ can be matched between the two systems with a choice of $j_\s$ that follows from the same corollary. However, the choice is also directly visible from the exact hypervirial relation \cref{eq:sigma-xi-j-relation}, which gives
\begin{subequations}
    \begin{align} \label{eq:j-relation}
        &j = -\omega^2\xi - g\sigma \\
        &j_\s = -\omega^2\xi.
    \end{align}
\end{subequations}
Consequently, the $\xi$ from the coupled system can always exactly be reproduced by choosing the above value for $j_s$ in the uncoupled system.
Now, let us in analogy to standard DFT define the direct-coupling and exchange-correlation potential as $\vdxc = v_\s - v$. From subtraction of Eqs.~\eqref{eq:FB-full-aux} we find that
\begin{equation}\label{eq:vdxc}
    \vdxc = \frac{t\sigma+g\ev{\hat\sigma_x\qop}{\psi}}{\ev{\hat\sigma_x}{\psi}} - \frac{t\sigma}{\ev{\hat\sigma_x}{\phi}}.
\end{equation}
This means any approximation of the ground-state $\psi$ of the coupled system will approximate the $\psi$-dependent parts of \cref{eq:vdxc}. This gives a functional approximation to use for the Kohn--Sham system. The arguably simplest approximation is $\psi=\phi$, the mean-field approximation, where $\ev{\hat\sigma_x\qop}{\phi} = \ev{\hat\sigma_x}{\phi}\ev{\qop}{\phi} = \ev{\hat\sigma_x}{\phi}\xi$, since the system is uncoupled and thus matter and photon parts factorize. On the other hand, this approximation misses correlation effects. From \cref{eq:vdxc}, we then have
\begin{equation*}
    \vdx = \vd = g\xi = \pdv{\sigma}D(\sigma,\xi).
\end{equation*}
This is just the direct-coupling part, that originates from the coupling term of the Hamiltonian $g\hat\sigma_z\qop \approx D(\sigma,\xi) = g\sigma\xi = \vd\sigma$. This means the exchange-only part vanishes, $\vx=0$, as was already seen in \cref{thrm:XC}.

We are, however, not limited to this level of approximation and can include some correlation information by using the adiabatic approximation on the level of quantum fluctuations. For a general operator we set $\hat A = \ev*{\hat A} + \Delta\hat A$, where $\Delta\hat A$ describes the (operator-valued) fluctuations around a mean value that are assumed to be ``small'', especially with respect to variation in time (adiabatic approximation). For the displacement operator this means following \cref{eq:dt-x} that $\Delta\qop \approx -\frac{g}{\omega^2} \Delta\hat\sigma_z$ and consequently that
\begin{equation}\label{eq:x-pf}
    \qop \approx \qop_\pf = \xi - \frac{g}{\omega^2}(\hat\sigma_z-\sigma).
\end{equation}
The same relation can be derived directly from \cref{eq:dt-x} by setting $\dtime\hat p=0$ and eliminating $j$ by inserting the exact relation from \cref{eq:j-relation}. With the help of \cref{eq:x-pf} the photon degree-of-freedom represented by $\qop$ can be effectively replaced by matter quantities and we can get an approximate ``photon-free'' formulation of the problem. This program was already put into effect in QEDFT and forms the basis for one of the first functionals for the Pauli--Fierz Hamiltonian in dipole approximation~\cite{schafer2021making}. Inserting this photon-free approximation into \cref{eq:vdxc} and setting $\psi=\phi$ yields the functional
\begin{equation}\label{eq:vdxc-eff}
    \vdc^\pf = g\xi + \frac{g^2}{\omega^2}\sigma - \underbrace{\frac{g^2}{\omega^2}\frac{\ev{\hat\sigma_x\hat\sigma_z}{\phi}}{\ev{\hat\sigma_x}{\phi}}}_{=0}.
\end{equation}
We realize quickly that the last term vanishes, since $\hat\sigma_x\hat\sigma_z=-i\hat\sigma_y$ and $\ev{\hat\sigma_y}=0$ for any eigenstate by \cref{eq:ev-sigmay}. This is beneficial since $i\hat\sigma_y$ is skew-adjoint and would thus have imaginary expectation values. We have thus derived an effective potential for a photon-free (uncoupled) system that aims at reproducing the same polarization $\sigma$ in both systems. A more detailed perturbation-theory analysis including the higher photon states shows that a factor $\eta_\mathrm{c} < 1$ should be introduced to take the matter-photon correlation into account, where $\eta_\mathrm{c}\to 1$ in the strong-coupling regime~\cite{lu2024electron}.
\begin{equation*}
    \vdc^{\pf,\eta_\mathrm{c}} = g\xi + \eta_\mathrm{c}\frac{g^2}{\omega^2}\sigma
\end{equation*}
While we will not enter the theoretical details here, a numerical demonstration shows that this leads to a high level of agreement for $\sigma$ values that are not too close to the boundary and that the linear approximation gets more accurate for strong coupling, see \cref{fig:v_hxc_lambda}. The dependence of the correlation factor $\eta_\mathrm{c}$ on $g$ is further illustrated in \cref{fig:eta_vs_lambda}.

\begin{figure}[ht]
    \centering
    \includegraphics{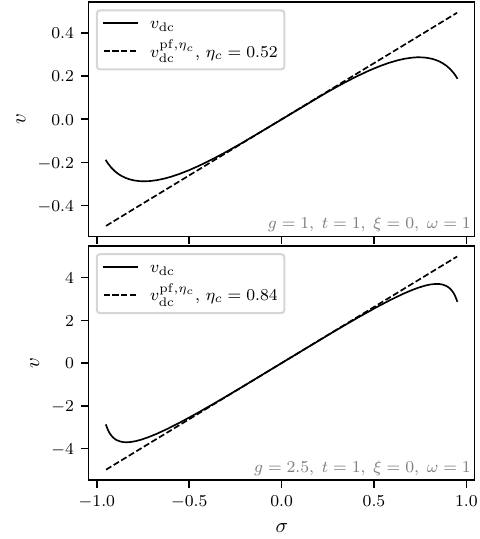}
  \caption{
  Comparison of the exact direct-coupling and correlation potential \( v_{\mathrm{dc}} \) and the photon-free approximation with correlation factor \( \eta_\mathrm{c} \), \( v_{\mathrm{dc}}^{\mathrm{pf}, \eta_\mathrm{c}} \), plotted as functions of the matter density \( \sigma \) for coupling constants \(g = 1 \) and \(g = 2.5 \). The approximation includes the correlation factor \( \eta_\mathrm{c} \) computed for each \(g \) from the derivative at $\sigma=0$. Parameters are set to \( \omega=1, t = 1, \xi = 0 \). The plots demonstrate the effectiveness of the photon-free approximation in reproducing the exact potential, especially not too close to the boundary and at higher coupling strengths.}
  \label{fig:v_hxc_lambda}
\end{figure}

\begin{figure}[ht]
  \centering
  \includegraphics{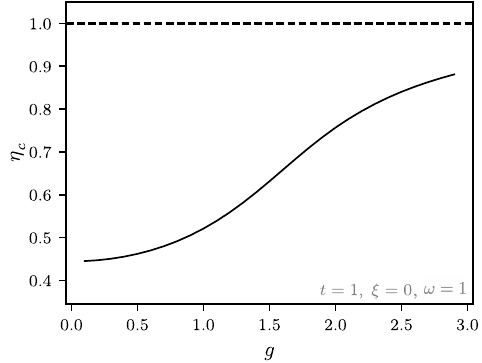}
  \caption{Dependence of the correlation factor \( \eta_\mathrm{c} \) on the coupling strength \(g \) for the quantum Rabi model. The plot illustrates how \( \eta_\mathrm{c} \) approaches 1 as \(g \) increases, indicating that the photon-free approximation with correlation factor \( \eta_c \) becomes exact in the strong-coupling limit for $\sigma$ not too close to the boundary. The calculations are performed with parameters \( \omega = 1, t = 1, \xi = 0 \). Different numerical cutoffs for maximal photon number are used for various ranges of \(g \) to ensure numerical accuracy.}
  \label{fig:eta_vs_lambda}
\end{figure}

We can check that this result fully matches the expression for the Levy--Lieb functional along the adiabatic connection from \cref{thrm:AC}. First take the difference between full and zero coupling to get the direct-coupling and correlation energy,
\begin{align*}
    E_\mathrm{dc}(\sigma,\xi) &= F_\mathrm{LL}^1(\sigma,\xi)-F_\mathrm{LL}^0(\sigma,\xi) \\
    &= g\sigma\xi -\frac{g^2}{2\omega^2}(1-\sigma^2) + I^1(\sigma).
\end{align*}
Since the potential corresponds to the negative differential of the respective functional, \cref{eq:F-subdiff}, we can directly use differentiation to get the external pair that encodes direct coupling and correlation for $\sigma\in(-1,1)$,
\begin{align*}
    &\vdc = v_\s - v = \frac{\partial E_\mathrm{dc}}{\partial\sigma} = g\xi + \frac{g^2}{\omega^2}\sigma + \frac{\partial I^1}{\partial\sigma},\\
    &j_\mathrm{d} = j_\s - j = \frac{\partial E_\mathrm{dc}}{\partial\xi} = g\sigma.
\end{align*}
This fits exactly to the previous results.
Note that \citet{Novokreschenov2023} recently gave an approximation for the effective potential based on diagrammatic expansion for the quantum Rabi model and the Dicke model.

The effective potential derived in this section aims at reproducing the value of the polarization $\sigma$ from the fully coupled system in a auxiliary system without light part. We can further ask how well (or if at all) other quantities can be reproduced as well, mainly the ground-state energy or the expectation value of parts of the Hamiltonian. This will be investigated in the next section by defining a whole photon-free Hamiltonian instead of just an effective potential.

\subsection{Photon-free Hamiltonian}
The previous, adiabatic approximation for operators allowed us to replace the photonic degree-of-freedom $\qop$ by pure matter quantities in \cref{eq:x-pf} and derive a correlation approximation. But we can go one step further and not only substitute this potential into the uncoupled Hamiltonian, but also provide a photon-free expression for the field energy. For this, remember that the ladder operator is
\begin{equation*}
    \hat a = \sqrt{\frac{\omega}{2}} \left( \qop + \frac{i}{\omega}\hat p \right).
\end{equation*}
We take the photon-free approximation for $\qop$ from \cref{eq:x-pf}. For the momentum operator we simply choose $\hat p_\pf = 0$ since we have $\hat p = \dtime\qop \approx 0$ from \cref{eq:dt-x} and the adiabatic approximation. This means we get $\hat a_\pf = \hat a_\pf^\dagger = \sqrt{\frac{\omega}{2}}\qop_\pf$, which is now a self-adjoint operator (since there are no more photons to create or annihilate). The photon energy is thus approximated by
\begin{equation*}
    \omega\left(\hat a^\dagger\hat a + \frac{1}{2}\right) \approx \omega\left(\hat a_\pf^2 + \frac{1}{2}\right) = \frac{\omega^2}{2}\qop_\pf^2 + \frac{\omega}{2}.
\end{equation*}
By insertion of \cref{eq:x-pf} for $\qop_\pf$ and using $\xi = -(g\sigma+j)/\omega^2$ from \cref{eq:j-relation}, we get a photon-free version of the quantum Rabi Hamiltonian of \cref{eq:FullHamiltonian},
\begin{equation*}
    \hat H_\pf(v,j) = -t\hat\sigma_x + \left( v - \frac{gj}{\omega^2}\right)\hat\sigma_z - \frac{j^2+g^2}{2\omega^2} + \frac{\omega}{2}.
\end{equation*}
We notice that the only difference to the Hamiltonian of the two-level system alone are the shift $-gj/\omega^2$ in the potential $v$, exactly what we derived before as $\vdc^\pf$ in \cref{eq:vdxc-eff} if we set $j = -\omega^2\xi-g\sigma$ again, and an overall shift in the energy. The spectra of the Hamiltonians and the difference in the expectation values of $\hat\sigma_z$ in the ground state are compared in \cref{fig:spectrum-photon-free}. Note that while the ground-state energy agrees only for small coupling, the expectation values of $\hat\sigma_z$ will also match again in the strong-coupling limit, where $\vdc^\pf$ gives the exact direct-coupling and correlation potential as discussed in \cref{subsec:eff-pot}.

\begin{figure}[ht]
  \centering
  \includegraphics{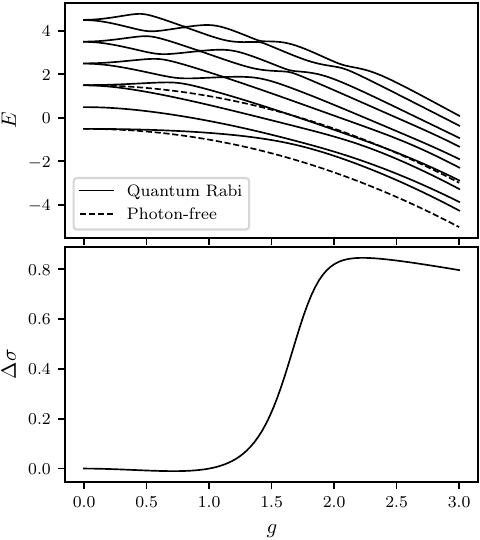}
  \caption{The upper plot shows the spectrum of the quantum Rabi Hamiltonian $\hat H(v,j)$ with $\omega=1,t=1,v=0.1, j=0.1$ at different couplings $g$ together with the spectrum of its photon-free approximation $\hat H_\pf(v,j)$ (dashed lines). The lower panel compares the expectation values of $\hat\sigma_z$ between the ground state $\psi$ of $\hat H(v,j)$ and $\psi_\pf$ of $\hat H_\pf(v,j)$, $\Delta\sigma=\ev{\hat\sigma_z}{\psi}-\ev{\hat\sigma_z}{\psi_\pf}$.}
  \label{fig:spectrum-photon-free}
\end{figure}

Finally, we point out that the photon-free approximation contains an entirely new ingredient when applied to the Dicke model. There, the effective potential at one site will include information about the polarization at other sites and thus amounts to an effective interaction between the two-level systems.

\section{Conclusions}\label{sec:Conclusions}
In summary, this work presents a thorough convex-analytical treatment of ground-state DFT applied to the most basic QED models, in particular the quantum Rabi model and the Dicke model. In the quantum Rabi model, not only can we achieve full $v$-representability for all polarizations (except $\sigma=\pm 1$) and displacements, but we are also able to derive quite explicit expressions for the adiabatic connection and the photon-free approximation. While many of the discussed features for a QEDFT will generalize from the minimal setting of the quantum Rabi model to more complex settings like the (multi-mode) Dicke model~\cite{Bakkestuen2024}, we do not have all results assured. Especially, it is an open question if the properties of the ground state from \cref{subsec:GSProperies} still hold for the Dicke model, and, as a derived property, if $v$-representability holds for the Dicke model.

With the quantum Rabi and Dicke models on solid grounds, the developed methodologies to analyze the properties of ground states and the corresponding QEDFT formulations can be extended to more challenging setups. While ultimately a characterization of the full (continuous) Pauli--Fierz problem would be desirable, already small modifications of the Dicke model might change the physics substantially. This is specifically so if we re-interpret the Dicke Hamiltonian in the length-gauge~\cite{rokaj2018light, schaefer2020relevance}. The length gauge is a unitary transformation of the Pauli--Fierz Hamiltonian in Coulomb gauge under the long-wavelength approximation, and mixes light and matter degrees of freedom in its basic coordinates~\cite{Ruggenthaler2023}. In this case, the Dicke Hamiltonian gets modified and already includes direct dipole-dipole interactions between the different two-level systems. This type of interaction is distinct from the usual Coulomb interaction and is currently assumed to be responsible for many of the observed effects in polaritonic chemistry for macroscopic ensembles of molecules under vibrational strong coupling~\cite{schnappinger2023,sidler2024unraveling}. Investigating this setting for the zero-temperature case as well as in thermal equilibrium is an obvious and very timely next step. Here, an important physical difference between the usual electronic DFT and many problems addressed with QEDFT becomes apparent: Due to a photonic structure that enhances certain photonic modes, different length-scales begin to talk to each other and we can no longer consider different molecules as statistically independent, not even in the dilute-gas limit. Thus QEDFT methodologies become important even for the case of simplified matter sub-systems, such as just two-level systems, because of the large number of coupled matter systems in macroscopic ensembles.

\acknowledgments
    Some of the authors have had the privilege of having Trygve Helgaker as a teacher and mentor in the field of mathematical DFT. It is therefore with great honor that we submitted this work to the celebration of Trygve's 70th birthday and presented a QEDFT that also highlights Trygve's viewpoint of DFT as being more of a discovery than an invention.
    VHB, ML, MP, and AL were supported by ERC-2021-STG grant agreement No. 101041487 REGAL. AL was also supported by Research Council of Norway through funding of the CoE Hylleraas Centre for Quantum Molecular Sciences Grant No. 262695 and CCerror Grant No.\ 287906. The authors would like to thank Mihály A.\ Csirik for fruitful discussions.

\section*{Data Availability}
    The source code for performing the numerical investigation is available on GitHub: \url{https://github.com/VegardFalmaar/QEDFT-Quantum-Rabi-Code}.


\section*{References}

%

\end{document}